\documentclass[aps,prl,preprint,a4paper]{revtex4-1}
\usepackage{cancel}
\usepackage{amsmath}
\usepackage{amssymb}
\usepackage{amsthm}
\usepackage{graphicx}% Include figure files
\usepackage{dcolumn}% Align table columns on decimal point
\usepackage{bm}% bold math
\usepackage{upgreek}
\usepackage{esint}
\usepackage{subcaption}
\usepackage{xcolor}

\theoremstyle{definition}

\theoremstyle{remark}
\newtheorem{remark}{Remark}

\theoremstyle{plain}
\newtheorem{theorem}{Theorem}

\theoremstyle{plain}

\theoremstyle{plain}

\theoremstyle{plain}

\begin{document}

%\preprint{}

\title{Non-perturbative guiding center model for magnetized plasmas}
%\title{Perturbative guiding center theory fails for $\alpha$-particles in stellarators} %Title of paper



\author{J. W. Burby}
\affiliation{Department of Physics and Institute for Fusion Studies, The University of Texas at Austin, Austin, TX 78712, USA}
\author{I. A. Maldonado}
\affiliation{Department of Physics, The University of Texas at Austin, Austin, TX 78712, USA}
\author{M. Ruth}
\affiliation{Department of Physics and The Oden Institute for Computational Engineering and Sciences, The University of Texas at Austin, Austin, TX 78712, USA}
\author{D. A. Messenger}
\affiliation{Theoretical Division, Los Alamos National Laboratory, Los Alamos, NM, 87545, USA}
\author{L. Carbajal}
\affiliation{Type One Energy Group Inc., Knoxville, TN 37931, USA}
 %\affiliation{New York University, New York, New York 10012, USA}
%\author{A. Cerfon}
%\affiliation{Courant Institute of Mathematical Sciences, New York, New York 10012, USA}
%\thanks{}


\date{\today}

\begin{abstract}
Perturbative guiding center theory adequately describes the slow drift motion of charged particles in the strongly-magnetized regime characteristic of thermal particle populations in various magnetic fusion devices. However, it breaks down for particles with large enough energy. We report on a data-driven method for learning a non-perturbative guiding center model from full-orbit particle simulation data. We show the data-driven model significantly outperforms traditional asymptotic theory in magnetization regimes appropriate for fusion-born $\alpha$-particles in stellarators, thus opening the door to non-perturbative guiding center calculations.
\end{abstract}


\maketitle %\maketitle must follow title, authors, abstract and \pacs

%In contrast to tokamaks, which rely on strongly self-organized plasma states for confinement, stellarators achieve confinement predominantly through application of external magnetic fields generated by highly optimized three-dimensional current-carrying coils. This confinement method affords plasma relatively few opportunities to tap free energy sources that lead to deleterious instabilities. But greater stability comes at the price of complicated particle confinement theory. The theory is so complex that early stellarators failed to compete with their tokamak counterparts. Modern understanding of quasisymmetric \cite{Bo2_1983,NZ_1988,Burby_phase_2013,Helander_review_2014,BKM_2020,Rodriguez_2020}, omnigeneous \cite{Hall_1975,Helander_review_2014,Landreman_2012,Parra_2015}, quasi-isodynamic \cite{QI_2023}, and isodrastic \cite{BMN_2023} magnetic fields provides practical optimization metrics that lead to stellarator designs with confinement quality for \emph{thermal plasma} comparable to tokamaks. However, as we will show, these metrics cannot be trusted for $3.5\text{ MeV}$ fusion-born $\alpha$-particles.

{ Successful design of a stellarator reactor requires optimizing fusion-born $\alpha$-particle confinement to sustain fusion burn and minimize damage to plasma facing components. Stellarator design efforts therefore require an efficient model for $\alpha$-particle dynamics.
State-of-the-art stellarator optimization routines employ the traditional guiding center model \cite{Kruskal_1962,Littlejohn_1981,Littlejohn_1983,Littlejohn_1984,Cary_2009,Burby_loops_JMP}, which eliminates the short cyclotron timescale and provides the theoretical basis for advanced confinement concepts such as quasisymmetry \cite{Bo2_1983,NZ_1988,Burby_phase_2013,Helander_review_2014,BKM_2020,Rodriguez_2020}, omnigeneity \cite{Hall_1975,Helander_review_2014,Landreman_2012,Parra_2015}, quasi-isodynamicity \cite{QI_2023}, and isodrasticity \cite{BMN_2023}. This approximation assumes $\epsilon = \rho_0/L_0\ll 1$, where $\rho_0$ denotes a characteristic particle gyroradius and $L_0$ denotes a characteristic length scale for the magnetic field $\bm{B}$. When $\epsilon$ is \emph{not} sufficiently small, traditional guiding center predictions may deviate significantly from those of full-Lorentz-force modeling, leading to imprecise estimates of particle confinement. Because the gyroradius scales linearly with particle speed, energetic $\alpha$-particles may fall in this problematic range even when thermal particles do not.  Accuracy of the traditional guiding center model should be critically evaluated {  in any future stellarator design study.}} 
% \mr{It's not quasi-isodynamicity?}
% \begin{figure}[h]
% \includegraphics[width=\linewidth]{eps_plot.png}
% \caption{
% Fusion-born $\alpha$-particle $\epsilon = \rho_0/L_0$ vs magnetic field strength $B_0$ in a device with scale length $L_0 = 1\,\text{m}$
% }
% \label{fig:eps_vs_B}
% \end{figure}

\begin{figure*}[htpb]
	\centering
    \includegraphics[width=1.0\linewidth]{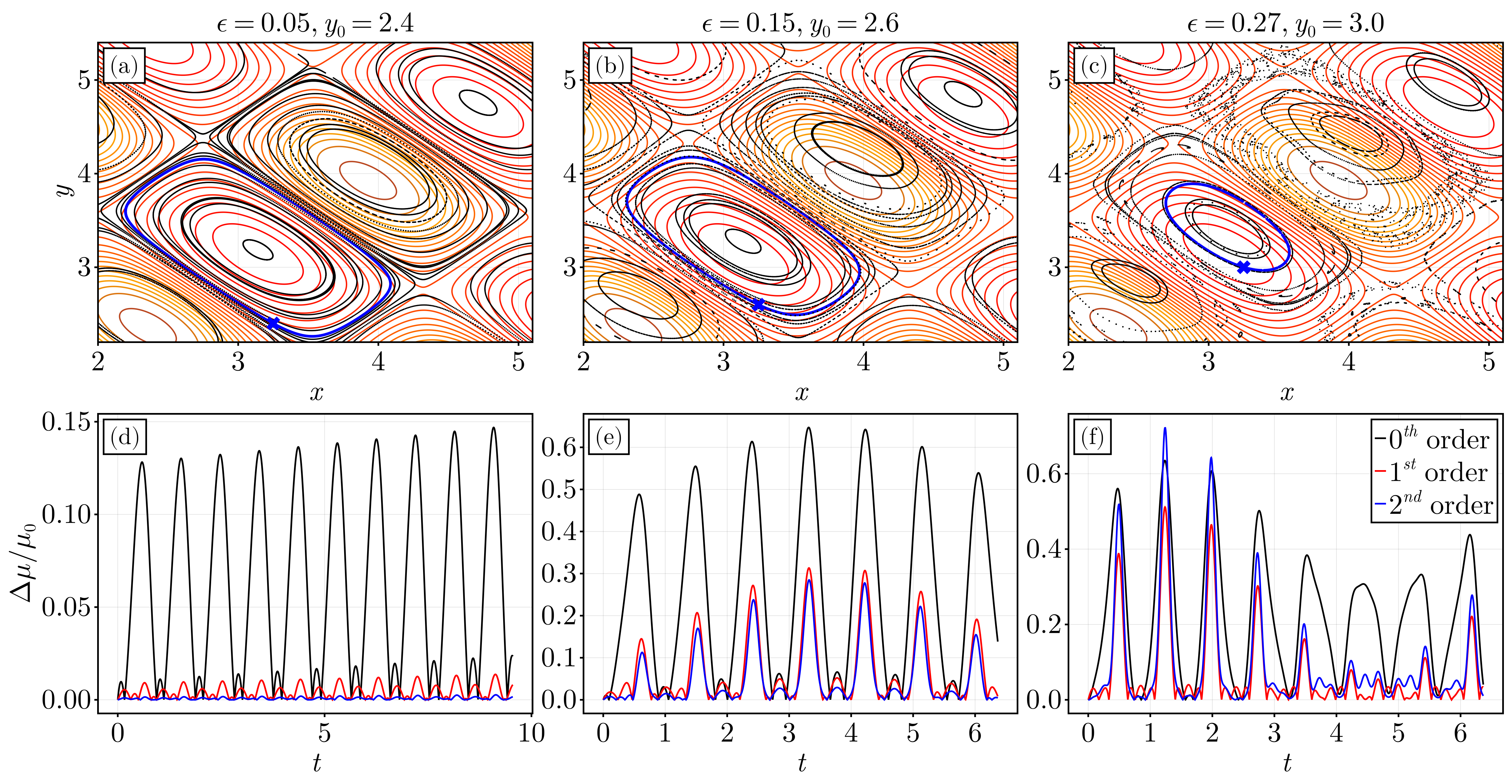}
	\caption{ (Top) Poincar\'e sections ($z=v_z=v_y=0, v_x>0$) of the full-orbit dynamics of Eq.~\eqref{eq:B} (black) and level sets of an adiabatic invariant truncation (color) {  initialized with $v_x=\epsilon$ for different values of $\epsilon$. Chaos appears in lobes with smaller field strength due to larger local-$\epsilon$.} { Lowest-order truncation is used in (a-c), consistent with guiding center models used in practical calculations. $y_0$ denotes $y$-coordinate of blue cross.} (Bottom) Relative error time traces of three truncations of adiabatic invariant series for the same $\epsilon$ values. Blue crosses denote trajectories underlying (d-f). 
 }
	\label{fig:epsilons}
\end{figure*}

% All of advanced stellarator confinement theory assumes that an asymptotic expansion, known as the guiding center model \cite{Kruskal_1962,Littlejohn_1981,Littlejohn_1983,Littlejohn_1984,Cary_2009,Burby_loops_JMP}, adequately describes dynamics of any given plasma particle. The dimensionless parameter $\epsilon = \rho_0/L_0$, equal to the ratio of a { typical} particle's gyroradius $\rho_0$ to the scale length of the magnetic field $L_0$, measures the quality of this assumption; smaller values of $\epsilon$ imply greater accuracy. The gyroradius of a fusion-born $\alpha$ is 19.6 times larger than that of a triton in $10\text{ keV}$ burning magnetized plasma. Thus, traditional guiding center theory describes $\alpha$-dynamics less accurately than thermal particle dynamics. While this observation alone need not cause concern, numerical study of $\alpha$-particle trajectories with $\epsilon$ in the range $[0.05,0.27]$
% %shown in Fig. \ref{fig:eps_vs_B} 
% suggests a serious issue as well as an exciting theoretical opportunity. 

{ Recent numerical studies \cite{Assuncao_2023,Rodrigues_2024,Carbajal_2024} \footnote{Ref.\,\cite{Carbajal_2024} is available upon request. Contact: L. Carbajal, \texttt{leo.carbajal@typeoneenergy.com}.} find significant deviations in computed $\alpha$-particle loss fractions in various reactor-scale stellarator configurations when using full-orbit vs traditional guiding center particle pushers. {{ These and other calculations} indicate that, depending on the stellarator configuration under study, guiding center predictions can either overestimate or underestimate the energetic particle loss fraction computed from full-orbit computations. While some configurations show no discrepancies at all, many lead to differences greater than $60\%$ of the guiding center prediction}. Quasisymmetric designs tend to show greater discrepancies than quasi-isodynamic designs. Single trajectory comparisons \cite{Assuncao_2023} reveal particles with velocity vectors nearly perpendicular to the magnetic field account for the bulk of these discrepancies. Current designs of some optimized stellarator fusion power plants \cite{Hegna_2025,Guttenfelder_2025} show good confinement of fusion-born alpha-particles with less than few percent of alpha-particle energy losses. Nevertheless, discrepancy between full-orbit and guiding-center estimates of alpha-particle confinement is still observed \cite{Carbajal_2025}, suggesting improvements over guiding center are still necessary for high-fidelity evaluations.} 

{ Such deviations may appear surprising because estimating $\epsilon$ for $\alpha$-particles in existing stellarator reactor concepts using nominal minor radius as magnetic scale length $L_0$ leads to $\epsilon\sim 2\text{-}3\%$. However, modern stellarator designs often feature highly elongated and curved poloidal cross sections. The \emph{local} $\epsilon$ in those regions may exceed the nominal value by a factor of ${\sim}\,5$, leading to a bleaker view of guiding center validity. 
For example, detailed phase portraits of $\alpha$-particle trajectories in the $2D$ magnetic field $\bm{B} = B(x,y)\,\bm{e}_z$, with
\begin{gather}
\label{eq:B}
    B(x,y)  = 1 + \sum_{i=1}^2a_i\,\cos(k_{xi}\,x + k_{yi}\,y),\\
\nonumber
    k_{x1} =k_{y2} =3,\quad k_{y1} = k_{x2} = 1,\quad 
    a_1 = a_2 = 0.3,
\end{gather}
reveal significant shifts in particle drift surfaces away from traditional guiding center predictions for $\epsilon \sim 15\%$ (note that $B$ is scaled so $L_0=1$).  See Fig.~\ref{fig:epsilons}. Since bounce motion and field line curvature play no role in this field, the quality of traditional guiding center predictions at $\epsilon\sim 15\%$ should only degrade further in a fully $3D$ configuration. Worse still, Fig.~\ref{fig:epsilons} reveals that including higher-order perturbative corrections \cite{Burby_gc_2013} to the traditional guiding center model, which represents a low-order truncation, offers limited accuracy improvements, if any. This signifies collapse of the optimal truncation order \cite{Berry_1991} for the adiabatic invariant asymptotic series. These observations motivate development of a non-perturbative variant of guiding center theory to replace traditional guiding center when it fails. } 

This Letter describes a non-perturabative guiding center model suitable for $\alpha$-particles in stellarators,{  runaway electrons in tokamaks \cite{Carbajal_2017,Carbajal_2020}, and magnetically-confined energetic particles \cite{Ogawa_2016} produced by neutral beams \cite{Pefferle_2015}, ion cyclotron resonance heating, or fusion reactions \cite{Cecconello_2018} more broadly.} { First we deduce formally-exact non-perturbative guiding center equations of motion assuming a hidden symmetry with associated conserved quantity $\mathcal{J}$. We refer to $\mathcal{J}$ as the non-perturbative adiabatic invariant.} These equations are completely determined by first-order derivatives of $\mathcal{J}$ and the magnetic field. They enjoy Hamiltonian structure comparable with that of asymptotic theory \cite{Littlejohn_1981}. Then we describe a data-driven method for learning $\mathcal{J}$ from a dataset of full-orbit $\alpha$-particle trajectories. We apply this method to the $\alpha$-particle dynamics shown in Fig.~\ref{fig:epsilons} and find the learned non-perturbative guiding center model significantly outperforms the standard guiding center expansion.  { Our proposed method for learning $\mathcal{J}$ applies on a per-magnetic field basis; changing $\bm{B}$ requires re-training. This makes it well-suited to stellarator design assessment tasks, such as $\alpha$-loss-fraction uncertainty quantification.} %\mr{This addition feels out of place to me, but I don't know where to put it.}

% Our work establishes the need for a non-perturbative adiabatic invariant that can be applied in any magnetic field configuration; the method of finding $\mathcal{J}$ reported here can only be applied on a per-magnetic-field basis. The fact that such a general-purpose adiabatic invariant exists in the perturbative regime suggests a non-perturbative general-purpose adiabatic invariant may exist as well.

In a seminal paper \cite{Kruskal_1962}, M. Kruskal showed traditional guiding center theory originates from a hidden perturbative $U(1)$-symmetry of the ordinary differential equations describing single-particle motion, $X = \dot{\bm{x}}\cdot\partial_{\bm{x}} + \dot{\bm{v}}\cdot\partial_{\bm{v}}$, where $\dot{\bm{v}} = \bm{v}\times\bm{B} $ and $\dot{\bm{x}} = \epsilon\,\bm{v}$. { Here $U(1)$-symmetry refers to a periodic one-parameter group of phase space transformations that maps solutions of the Lorentz force equations into other solutions. At leading order, the $U(1)$-symmetry is concretely expressed via gyrophase translations.} Kruskal dubbed the infinitesimal generator of the perturbative hidden symmetry the roto rate $R$ and showed that its formal expansion in powers of $\epsilon$ is uniquely determined to all orders, with first term given by $R_0 = \bm{v}\times\bm{b}\cdot\partial_{\bm{v}}$. 

Our non-perturbative guiding center model assumes (I) existence of a non-perturbative $U(1)$-symmetry with infinitesimal generator $\mathcal{R}$ such that $\mathcal{R}|_{\epsilon = 0} = R_0$. We refer to streamlines of $\mathcal{R}$ as $U(1)$-orbits. Here symmetry means the non-perturbative roto rate $\mathcal{R}$ is a Hamiltonian vector field with Hamiltonian function $\mathcal{J}$ that commutes with the kinetic energy $E = |\bm{v}|^2/2$ under the Lorentz force Poisson bracket,
\begin{align*}
\{F,G\} = \bm{B}\cdot \partial_{\bm{v}}F\times\partial_{\bm{v}}G + \epsilon\,(\partial_{\bm{x}}F\cdot\partial_{\bm{v}}G -\partial_{\bm{x}}G\cdot\partial_{\bm{v}}F),
\end{align*}
i.e. $\{E,\mathcal{J}\} = 0$. We refer to $\mathcal{J}$ as the non-perturbative adiabatic invariant.

To ensure topological similarity with the known $U(1)$-symmetry at $\epsilon = 0$, we also assume (II) each $U(1)$-orbit generated by $\mathcal{R}$ uniquely and transversally intersects a certain phase space Poincar\'e section $\Sigma$ defined as follows. Introduce phase space coordinates $(\bm{x},v_\parallel,v_\perp,\zeta)$ such that $\bm{v} = v_\parallel\,\bm{b} + v_\perp\,(\cos\zeta\bm{e}_1 + \sin\zeta\bm{e}_2)$. Here $\bm{e}_1,\bm{e}_2$ are unit vector fields chosen to ensure $(\bm{e}_1,\bm{e}_2,\bm{b})$ is a right-handed orthonormal frame. Then set $\Sigma = \{(\bm{x},v_\parallel,v_\perp,0)\}$. We refer to $\Sigma$ as the guiding center Poincar\'e section. 

%Since $\mathcal{R}|_{\epsilon=0} = R_0 = -\partial_\zeta$, (I) $\Rightarrow$ (II) when $\epsilon = 0$. Due to persistence of transversal intersections under deformations \cite{GP_1974},  (I) $\Rightarrow$ (II) for $\epsilon$ in some open interval $(-\epsilon_0,\epsilon_0)$, $\epsilon_0 > 0$, and $v_\perp$ not too close to $0$. However, (II) is still essential because we cannot predict $\epsilon_0$ in advance.  

Intrinsically, a particle's guiding center is equal to the $U(1)$-orbit passing through that particle's phase space location. Guiding center phase space is therefore intrinsically the collection of all $U(1)$-orbits. Constructing a guiding center model, perturbative or non-perturbative, requires equipping this space with coordinates and identifying an evolution law for $U(1)$-orbits in that coordinate system. Traditional guiding center theory painstakingly constructs such coordinates order-by-order in $\epsilon$, resulting in non-unique model equations \cite{Parra_Calvo_2014} with exploding complexity that cannot be applied in the non-perturbative regime identified in Fig. \ref{fig:epsilons}. 

Breaking with tradition, we label $U(1)$-orbits by assigning to each orbit its point of intersection with $\Sigma$, $(\bm{x},v_\parallel,v_\perp,\zeta) = (\bm{X},u,r,0)$, or footpoint.  In this manner we regard $\Sigma$ as a concrete realization of the abstract space of $U(1)$-orbits. A particle's footpoint is just the footpoint of the $U(1)$-orbit passing through that particle's phase space location. We denote the footpoint of a particle at $(\bm{x},v_\parallel,v_\perp,\zeta)$ as $\pi(\bm{x},v_\parallel,v_\perp,\zeta) = (\bm{X},u,r)$. Note that $\pi(\bm{x},v_\parallel,v_\perp,0) = (\bm{x},v_\parallel,v_\perp)$, indicating that a particle's footpoint coincides with its phase space location whenever the particle's trajectory intersects $\Sigma$, which happens once each gyroperiod. We define the guiding center of any particle as that particle's footpoint on $\Sigma$.

% The $U(1)$-orbit passing through a particle-space point $(\bm{x},v_\parallel,v_\perp,\zeta)$

% We use $\pi$ to denote the footpoint map: it assigns to each particle-space point $(\bm{x},v_\parallel,v_\perp,\zeta)$ its $U(1)$-orbit-mate on $\Sigma$

% We use $\pi$ to denote the map that sends each particle-space point $(\bm{x},v_\parallel,v_\perp,\zeta)$ to its $U(1)$-orbit-mate on $\Sigma$, $\pi(\bm{x},v_\parallel,v_\perp,\zeta) = (\bm{X},u,r)$.  We refer to $\pi$ as the footpoint map and the image of $(\bm{x},v_\parallel,v_\perp,\zeta)$ under $\pi$ as that phase point's footpoint. Note that $\pi(\bm{x},v_\parallel,v_\perp,0) = (\bm{x},v_\parallel,v_\perp)$, indicating that a particle's footpoint coincides with its phase space location whenever the particle's trajectory intersects $\Sigma$. We define the guiding center of any particle as that particle's footpoint on $\Sigma$.

%Even in the perturbative regime, $\epsilon\ll 1$, our definition of the guiding center differs from traditional definitions at first order in perturbation theory, as shown in Fig. XXX. 

\begin{figure*}[htpb]
	\centering
        \includegraphics[width=\linewidth]{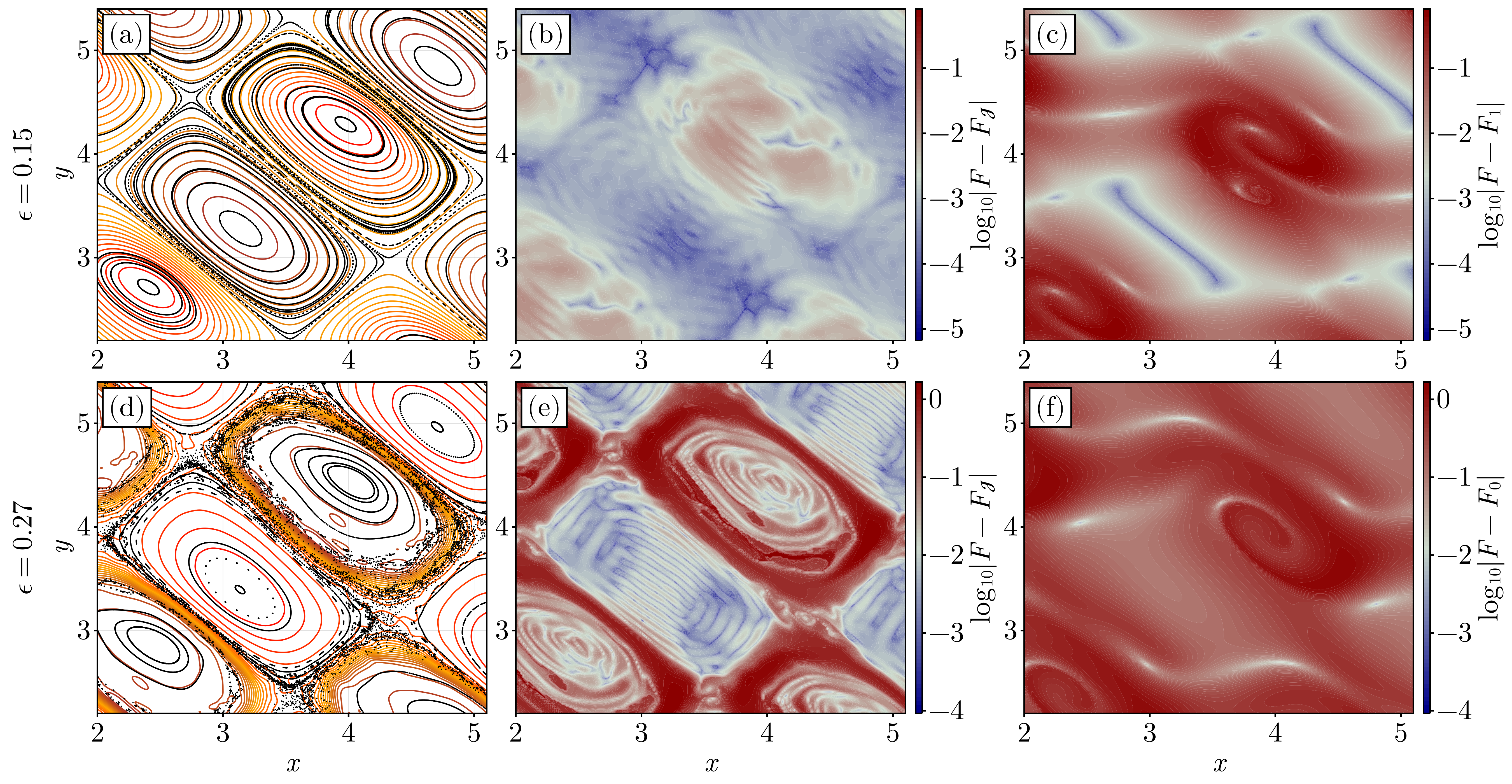}
	\caption{(a,d) Poincar\'e sections of the full-orbit dynamics (black) and level sets of the learned $\mathcal{J}_{\Sigma}$ (color) for $\epsilon \in \{0.15,0.27\}$ (cf.~Fig.~\ref{fig:epsilons} (b,c)). (b,e) Error of $\mathcal{J}_\Sigma$ prediction of the dynamics from the full-orbit dynamics over a gyroperiod. (c,f) Error of the order-$j$ truncated dynamics \eqref{perturb_Jsig} over a gyroperiod with $j=(1,0)$. { The error is computed on a uniform $201\times201$ point grid --- a higher resolution than the training data.}}
	\label{fig:Js}
\end{figure*}

As a particle moves through phase space its footpoint traces a curve on $\Sigma$. Because the Lorentz force equations of motion are $U(1)$-invariant by hypothesis that curve is a streamline for a uniquely defined vector field on $\Sigma$ that we call the footpoint flow field, $X_\Sigma = \dot{\bm{X}}\cdot\partial_{\bm{X}} + \dot{u}\,\partial_u + \dot{r}\,\partial_r$. { For a proof, see Theorem 1 in supplementary material.} 
%By identifying $\Sigma$ with the space of $U(1)$-orbits we identify footpoint dynamics with $U(1)$-orbit dynamics as a byproduct.  
%Thus, 
The components of the footpoint flow field define our non-perturbative guiding center equations of motion. We identify explicit formulas for these components involving only $\mathcal{J}$, $\bm{B}$, and the unit vectors $\bm{e}_1,\bm{e}_2$ as follows. { See Figure 1 in the supplementary material for an illustration of the phase space geometry of footpoint flow}.

Any pair of functions $f,g$ on $\Sigma$ determines a pair of $U(1)$-invariant functions $F = f\circ\pi$, $G = g\circ\pi$ on all of phase space by precomposing with the footpoint map $\pi$. Because $\mathcal{R}$ is a Hamiltonian vector field the function $\{F,G\}$ is also $U(1)$-invariant. It follows that the restriction $\{f,g\}_\Sigma = \{F,G\}\mid\Sigma$ defines a Poisson bracket on $\Sigma$. This bracket, together with the restricted kinetic energy, $E_\Sigma = E\mid \Sigma = u^2/2 + r^2/2$, determines a Hamiltonian system on $\Sigma$. {Theorem 2 in supplementary material shows this system coincides with the footpoint flow field.} Thus, components of the footpoint flow field $X_\Sigma$ may be computed explicitly given an explicit formula for the Poisson bracket $\{f,g\}_\Sigma$. This bracket can be expressed entirely in terms of $\mathcal{J}$, $\bm{B}$, and the unit vectors $\bm{e}_1,\bm{e}_2$. The general formula is contained in the proof of Theorem 2 in supplementary material, which includes Ref.\,\cite{Guillemin_Sternberg_1980}. For the $4D$ $(x,y,v_\perp,\zeta)$-phase space appropriate for magnetic fields of the form $\bm{B} = B(x,y)\,\bm{e}_z$ the result is 
\begin{align*}
        \{f,g\}_\Sigma &=\epsilon\,(\partial_{X} f\,\partial_{r}g -\partial_{r}f\,\partial_{X} g )\\
        & + \epsilon\,(\partial_{X}\mathcal{J}_\Sigma \partial_{r}f - \partial_{r}\mathcal{J}_\Sigma \partial_{X}f)\frac{ \partial_{r}g + \epsilon\,B^{-1}\, \partial_{Y}g}{  \partial_{r}\mathcal{J}_\Sigma+\epsilon\,B^{-1}\,\partial_{Y}\mathcal{J}_\Sigma}\\
        & - \epsilon\,(\partial_{X}\mathcal{J}_\Sigma \partial_{r}g - \partial_{r}\mathcal{J}_\Sigma \partial_{X}g)\frac{  \partial_{r}f+\epsilon\,B^{-1}\, \partial_{Y}f}{  \partial_{r}\mathcal{J}_\Sigma+\epsilon\,B^{-1}\,\partial_{Y}\mathcal{J}_\Sigma}.
    %     \epsilon \bigg[(\partial_{X} f\,\partial_{r}g -\partial_{r}f\,\partial_{X} g ) + (\partial_{X}\mathcal{J}^\Sigma_\epsilon \partial_{r}f - \partial_{r}\mathcal{J}^\Sigma_\epsilon \partial_{X}f)\frac{\epsilon \partial_{Y}g + |B|\partial_{r}g}{\epsilon\partial_{Y}\mathcal{J}^\Sigma_\epsilon + |B|\partial_{r}\mathcal{J}^\Sigma_\epsilon}\\
    % &-(\partial_{X}\mathcal{J}^\Sigma_\epsilon \partial_{r}g - \partial_{r}\mathcal{J}^\Sigma_\epsilon \partial_{X}g)\frac{\epsilon \partial_{Y}f + |B|\partial_{r}f}{\epsilon\partial_{Y}\mathcal{J}^\Sigma_\epsilon + |B|\partial_{r}\mathcal{J}^\Sigma_\epsilon}\bigg]
\end{align*}

{Theorem 3 in supplementary material gives explicit formulas for the non-perturbative guiding center equations of motion for general $\bm{B}$.} For $\bm{B}= B(x,y)\,\bm{e}_z$, the results simplify to $\dot{X} = U/D, \dot{Y} = V/D, \dot{r} = 0$, where
\begin{gather}
U = -\epsilon^2\,\partial_{Y}\mathcal{J}_\Sigma,\quad V = \epsilon^2\,\partial_{X}\mathcal{J}_\Sigma,\label{NPGC_UV}\\
D = (B/r)\,(\partial_{r}\mathcal{J}_\Sigma+\epsilon\,B^{-1}\, \partial_{Y}\mathcal{J}_\Sigma).\label{NPGC_D}
% \dot{X} & =\phantom{-}\frac{\epsilon^2}{B}\frac{r\,\partial_{Y}\mathcal{J}^\Sigma_\epsilon }{\partial_{r}\mathcal{J}^\Sigma_\epsilon+\epsilon\,B^{-1}\, \partial_{Y}\mathcal{J}^\Sigma_\epsilon }
% \label{NPGC_X}\\
% \dot{Y} & = -\frac{\epsilon^2}{B}\frac{r\,\partial_{X}\mathcal{J}^\Sigma_\epsilon }{\partial_{r}\mathcal{J}^\Sigma_\epsilon+\epsilon\,B^{-1}\, \partial_{Y}\mathcal{J}^\Sigma_\epsilon  }\label{NPGC_Y}\\
% \dot{r} & = 0.\label{NPGC_r}
\end{gather}
Here $\mathcal{J}_\Sigma = \mathcal{J}\mid \Sigma$ denotes restriction of the non-perturbative adiabatic invariant to the guiding center Poincar\'e section. These non-perturbative guiding center evolution laws are exact granted existence of the non-perturbative $U(1)$-symmetry generated by $\mathcal{R}$.

Although formulas \eqref{NPGC_UV}-\eqref{NPGC_D} reveal the central role played by the adiabatic invariant in guiding center modeling, they cannot be numerically simulated  without an expression for $\mathcal{J}_\Sigma$. In the perturbative regime, $\epsilon\ll 1$, truncations of the magnetic moment asymptotic series, such as the second-order result \cite{Burby_gc_2013} for $\bm{B} = B(x,y)\,\bm{e}_z$,
\begin{multline}
\label{perturb_Jsig}
    \mathcal{J}_\Sigma  = \frac{r^2}{2 B} - \epsilon\,\frac{r^3\,\partial_yB}{2B^3} + \epsilon^2 \frac{r^4}{16B^{5}} \left(3\left[\left(\partial_x B\right)^2+5\left(\partial_y B\right)^2\right]-B\left(\partial_x^2 B + 5\partial_y^2 B\right)\right),
\end{multline}
provide this missing ingredient. 
On the other hand, Fig.~\ref{fig:epsilons} demonstrates inadequacy of truncated series representations in the non-perturbative regime.

%in the non-perturbative regime identified in Fig.~\ref{fig:epsilons}, the traditional truncated series representation for $\mathcal{J}_{\Sigma}$ fails.  

We instead choose to learn $\mathcal{J}_{\Sigma}$ directly from trajectories of the Lorentz force equations for \eqref{eq:B} by minimizing the Rayleigh quotient
\begin{equation}
\label{eq:minJ}
    \min_{\mathcal{J}_{\Sigma}} \frac{R_{\mathrm{Dyn}}^2(\mathcal{J}_{\Sigma}) + R_{\mathrm{Inv}}^2(\mathcal{J}_{\Sigma})}{\left\lVert\mathcal{J}_{\Sigma}\right\rVert^2_{L^2(\Omega)}}, \quad \mathrm{s.t. } \int_{\Omega}\mathcal{J}_{\Sigma} \, \mathrm{d}^3\bm{r} = 0,
\end{equation}
where $\mathcal{J}_{\Sigma}$ is discretized by $39\times39\times2$ Fourier by Fourier by Chebyshev modes in $\bm{r} = (x,y,r) \in \Omega = \Omega_{xy}\times[\sqrt{2}-0.01, \sqrt{2}+0.01]$ and $\Omega_{xy}$ is a periodic cell of the magnetic field. Here $\lVert \mathcal{J}_{\Sigma} \rVert_{L^2(\Omega)}^2 = \frac{1}{N}\sum_i\mathcal{J}_{\Sigma}(\bm{r}_i)^2$,
where the $\bm{r}_i$ are sampled from a $79\times79\times3$ Fourier by Fourier by Chebyshev-Lobatto quadrature grid on $\Omega$ and $N=51^2\times9$. 
The first residual
\begin{multline*}
R_{\mathrm{Dyn}}^2 = \frac{1}{\hat{N}}\sum_i[(D(\hat{\bm{r}}_i;\mathcal{J}_{\Sigma})\,\dot{X}_i - U(\hat{\bm{r}}_i;\mathcal{J}_{\Sigma}))^2  +(D(\hat{\bm{r}}_i;\mathcal{J}_{\Sigma})\,\dot{Y}_i^n - V(\hat{\bm{r}}_i;\mathcal{J}_{\Sigma}))^2]
\end{multline*}
penalizes differences between predictions from $\mathcal{J}_{\Sigma}$ and full-orbit trajectories on the Poincar\'e section. A similar objective function appeared previously in \cite{Messenger_2024}, where a parametric averaged Hamiltonian appeared in place of our parametric $\mathcal{J}_{\Sigma}$. 
For the sum we initially sample $1000$ points $\hat{\bm{r}}_i$ from a low-discrepancy sequence on $\Omega$. 
For the $\hat{N}$ points that correspond to integrable trajectories, we estimate footpoint time derivatives $(\dot{X}_i,\dot{Y}_i)$ using a technique based on \cite{Ruth:2024}. (See supplementary materials.) 
The residual $R_{\mathrm{Dyn}}^2$ ignores non-integrable trajectories. The second residual is given by $R_{\mathrm{Inv}}^2 = \frac{1}{N}\sum_i(\mathcal{J}_{\Sigma}(\bm{r}_i)-\mathcal{J}_{\Sigma}(F(\bm{r}_i)))^2$, where $F : \Omega \to \Omega$ denotes the full-orbit Poincar\'e map on $\Sigma$. It penalizes non-invariance of $\mathcal{J}_{\Sigma}$ \cite{Ruth:2023}. This, along with the oversampling in $\bm{r}_i$, smoothes $\mathcal{J}_{\Sigma}$ in chaotic regions. The quotient can be minimized via a single generalized eigenvalue problem because both residuals in \eqref{eq:minJ} are quadratic in $\mathcal{J}_{\Sigma}$.
{ Computing the $\hat{N}$ samples of $({X},{Y})$ dominates computational time. An average trajectory computation requires around $1000$ gyroperiods, so the method's total cost is on the order of the time to simulate $1000$ gyroperiods for $1000$ particles using a full-orbit integrator.}

{ Many other methods compute approximate invariant structures, including Lagrangian coherent structures \cite{hadjighasem2017,froyland2014}, flux surfaces \cite{dewar-qfm-maps,dewar-qfm-fields,Ruth:2023}, and invariant tori more generally \cite{haro-parameterization}. However, these methods do not distinguish between the true adiabatic invariant $\mathcal{J}$ associated with continuation of Kruskal's perturbative $U(1)$-symmetry and arbitrary invariants of the form $I(E,\mathcal{J})$. Our method accounts for this ambiguity because replacing $\mathcal{J}$ with $I(E,\mathcal{J})$ in the non-perturbative equations of motion \eqref{NPGC_UV} generally changes predicted footpoint trajectories, and thereby increases the residual $R_{\mathrm{Dyn}}^2(\mathcal{J}_{\Sigma})$. In contrast to general constants of motion, the learned non-perturbative invariant $\mathcal{J}$ completely {determines the dynamics}.}

% To define the second residual, let $F : \Omega \to \Omega$ denote the full-orbit Poincar\'e map on $\Sigma$. 
% We define $R_{\mathrm{Inv}}^2 = \frac{1}{N}\sum_i(\mathcal{J}_{\Sigma}(\bm{r}_i)-\mathcal{J}_{\Sigma}(F(\bm{r}_i)))^2$ 
% % and $\lVert \mathcal{J}_{\Sigma} \rVert_{L^2(\Omega)}^2 = \frac{1}{N}\sum_i\mathcal{J}_{\Sigma}(\bm{r}_i)^2$,
% % where the $\bm{r}_i$ are sampled from a $79\times79\times3$ Fourier by Fourier by Chebyshev-Lobatto quadrature grid on $\Omega$ and $N=51^2\times9$. 
% The residual $R_{\mathrm{Inv}}$ is minimized when $\mathcal{J}_{\Sigma}$ is invariant \cite{Ruth:2023}. 
% This, along with the oversampling in $\bm{r}_i$, serves to smooth $\mathcal{J}_{\Sigma}$ in the chaotic regions.
% Note that both residuals in \eqref{eq:minJ} are quadratic in $\mathcal{J}_{\Sigma}$, so the quotient can be minimized via a single generalized eigenvalue problem. 
% also note that here we do not employ sparse regression techniques, as in SINDy \cite{Brunton_2016} and WSINDy \cite{Messenger_2021}.

%%%% Discussion
Fig.~\ref{fig:Js} compares predictions using learned $\mathcal{J}_\Sigma$ to full-orbit simulation data previously shown in Fig.~\ref{fig:epsilons} at $\epsilon \in \{0.15,0.27\}$. 
The global phase portraits in Fig.~\ref{fig:Js} (a,d) reveal significant improvements over the asymptotics shown in Fig.~\ref{fig:epsilons} (b,c). 
%Note the elimination of unphysical protrusions present in level sets of the perturbative adiabatic invariant visible in Fig.~\ref{fig:epsilons} (b,c). 
Let $F_{\mathcal{J}}$ denote the approximate Poincar\'e map obtained by evolving Eqs.~\eqref{NPGC_UV}-\eqref{NPGC_D} over a gyroperiod using learned $\mathcal{J}$. Let $F_j$ denote the corresponding map using order-$j$ truncation of \eqref{perturb_Jsig}. Figure \ref{fig:Js} (b,e) compares the log-absolute errors, $\log_{10} |F_{\mathcal{J}}-F|$
and $\log_{10}|F_j-F|$, with $j=1$ for $\epsilon=0.15$ and $j=0$ for $\epsilon=0.27$ (larger $j$ leads to zero denominators $D$ in both cases). For both $\epsilon$-values, Poincar\'e map dynamics predicted using learned $\mathcal{J}_\Sigma$ outperforms traditional asymptotics by orders of magnitude over most of the phase portrait.  

% is plotted in Fig.~\ref{fig:Js} (b,e), which can be compared to the log-absolute error $\log_{10}|F_j-F|$ with $j=1$ for $\epsilon=0.15$ and $j=0$ for $\epsilon=0.05$ (higher orders result in zero denominators $D$ in both cases). 
% For both values of $\epsilon$, the guiding center dynamics predicted using the learned $\mathcal{J}_\Sigma$ outperforms the traditional asymptotic theory by orders of magnitude over most of the phase portrait.  

% To quantitatively measure the error in the invariant, let $F_{\mathcal{J}}$ be the map obtained by evolving Eqs.~\eqref{NPGC_UV}-\eqref{NPGC_D} over a gyroperiod for the learned invariant, and $F_j$ be the equivalent map for the order-$j$ truncation of \eqref{perturb_Jsig}.
% The log-absolute error $\log_{10} |F_{\mathcal{J}}-F|$ is plotted in Fig.~\ref{fig:Js} (b,e), which can be compared to the log-absolute error $\log_{10}|F_j-F|$ with $j=1$ for $\epsilon=0.15$ and $j=0$ for $\epsilon=0.05$ (higher orders result in zero denominators $D$ in both cases). 
% For both values of $\epsilon$, the guiding center dynamics predicted using the learned $\mathcal{J}_\Sigma$ outperforms the traditional asymptotic theory by orders of magnitude over most of the phase portrait.   

{ %We anticipate at least three important future applications of non-perturbative guiding center theory: new advanced confinement principles, stellarator design certification, and stellarator design optimization. 

\emph{Discussion--} Advanced stellarator design principles like quasisymmetry and omnigeneity aim to improve confinement by imbuing the traditional guiding center model with desirable dynamical features, such as symmetry or small radial flux. Reliance on traditional guiding center limits applicability of these principles to the regime $\epsilon\ll 1$. New confinement principles targeting comparable dynamical improvements for the non-perturbative guiding center model promise to extend the benefits of optimized confinement to larger-$\epsilon$ particles, such as $\alpha$-particles in stellarators that break the traditional guiding center model. %For example, where quasisymmetry asks for symmetry of the guiding center equations of motion an extension could instead ask for symmetry of the non-perturbatibe guiding center equations of motion.

Assessing a candidate stellarator reactor design merits sensitivity analysis of $\alpha$-particle loss fraction over an energy confinement time. Assessing sensitivity to perturbations of the $\alpha$-particle distribution function requires either high-dimensional adjoint solves or large Monte-Carlo (MC) sampling of $\alpha$-particle trajectories from various perturbed distributions. Because each MC sample evolves in the same field configuration, computational overhead associated with training the non-perturbative guiding center model falls well below the cost of performing all required Monte-Carlo sampling using a full-orbit approach. Thus, in device designs where guiding center fails for $\alpha$-particles, accuracy and timestepping efficiency motivate trajectory sampling using the non-perturbative guiding center model over traditional guiding center or full-orbit models. %Benefits of the non-perturbative approach would include faster trajectory generation over full-orbit with higher orbit accuracy over traditional guiding center.

Direct optimization of $\alpha$-particle confinement \cite{albert_accelerated_2020,albert_alpha_2023,bindel_direct_2023} entails repeated ensemble trajectory simulations during an inner field-optimization loop. Computational feasibility of this technique currently requires traditional guiding center for trajectory calculations, limiting its applicability to particles with $\epsilon\ll 1$. Deploying a trained non-perturbative guiding center model in place of traditional guiding center promises to extend applicability of direct fast-particle confinement optimization to particles with non-perturbative $\epsilon$. However, the training method proposed here falls short for inner-loop applications because it must be repeated whenever the field changes, thereby encurring unacceptable computational overhead comparable to full-orbit trajectory generation. Overcoming this barrier requires solving the operator learning problem \cite{LuLu_2021} of mapping $\bm{B}$ to $\mathcal{J}$. Recent advances in machine learning, especially in foundation models \cite{Subramanian_2023} and sparse regression \cite{Brunton_2016,Messenger_2021}, stand poised to tackle this challenge.

%When guiding center fails for $\alpha$-particles, accuracy and timestepping efficiency motivates trajectory sampling using the non-perturbative guiding center model over traditional guiding center or full-orbit. The non-perturbative promises a practical win when the computational overhead of model training falls well below . The data-driven approach to constructing the non-perturbative guiding center model presented here promises to accelerate crucial uncertainty quantification (UQ) tasks like this one where the number of particle trajectories required for UQ far exceeds the number required for training the non-perturbative model in a single magnetic field. 

}

{ Most models for runaway electron trajectories in tokamaks also employ the traditional guiding center equations. However, loss fractions for runaway electrons in tokamak fields may deviate significantly \cite{Carbajal_2020} when computed using traditional guiding center or full-orbit modeling. Similarly, relative field variations along cyclotron orbits for $30\text{ MeV}$ runaways in DIII-D-like fields reach values in excess of $15\%$ \cite{Carbajal_2017}. After suitable modification accounting for relativistic effects, non-perturbative guiding center theory therefore offers the opportunity to correct runaway trajectory errors committed by traditional guiding center without resorting to full-orbit simulation.}

%{ Fast particles created by mechanisms other than runway or fusion reactions create additional opportunities for potentially fruitful applications of non-perturbative guiding center theory. For example, some fast particles trajectories in MAST configurations with helical cores \cite{Pefferle_2015} show qualitative differences between full-orbit and guiding center. The non-perturbative model applies to these and other large gyroradius magnetized particles more broadly.}

%{ Non-perturbative guiding center modeling applies  to runaway electrons and other fast particles in magnetic fusion devices. Guiding center predictions of runaway dynamics in ITER-like fields differ qualitatively for full-orbit predictions \cite{Carbajal_2017,Carbajal_2020}.  Fast particles trajectories in MAST configurations with helical cores \cite{Pefferle_2015} also show qualitative differences between full-orbit and guiding center.  }

%{ To mention in discussion: runaways, orbits in SPEC fields, fast particles in tokamaks. Refer to David's paper here \cite{Pefferle_2015} for evidence of application to tokamaks. ``Specific use cases" a la referee a.}

\noindent\emph{Acknowledgements--} This material is based on work supported by the U.S. Department of Energy, Office of Science, Office of Advanced Scientific Computing Research, as a part of the Mathematical Multifaceted Integrated Capability Centers program, under Award Number DE-SC0023164. It was also supported by U.S. Department of Energy grant \# DE-FG02-04ER54742, as well as the Los Alamos National Laboratory LDRD program under Project No. 20240858PRD.

%%% supplementary material

\section{Supplementary material}
\subsection{Dynamical systems underpinning}

This Section provides statements and proofs of some basic results in the theory of ordinary differential equations (ODEs) with $U(1)$-symmetry and Hamiltonian structure. It also shows how these results generalize the theory presented in the main text from magnetic fields of the special form $\bm{B} = B(x,y)\,\bm{e}_z$ to general non-vanishing magnetic fields. Throughout, $\dot{z} = X(z)$ denotes a system of first-order ODEs on the phase space $Z\ni z$. Note that $X$ may be understood as a vector field on $Z$ that encodes the ODE. We make no distinction between the ODE system and the vector field $X$. We will always use the symbol $F_t = \exp(t\,X):Z\rightarrow Z$ to denote the time-$t$ flow map for $X$. This discussion assumes smoothness of the various geometric objects that appear.

As is standard, the symbol $U(1)$ denotes the group of complex numbers with unit modulus $e^{i\theta}$. We identify this group with the set of real numbers $\theta$ modulo $2\pi$. To formalize the notion of $U(1)$-symmetry we refer to $U(1)$-actions. A $U(1)$-action on $Z$ is a family of mappings $\Phi_\theta:Z\rightarrow Z$, parameterized by $\theta\in U(1)$, such that $\Phi_0=\Phi_{2\pi} = \text{id}_Z$ and $\Phi_{\theta_1 + \theta_2} = \Phi_{\theta_1}\circ\Phi_{\theta_2}$ for each $\theta_1,\theta_2\in U(1)$. 

We think of a $U(1)$-action as a collection of generalized rotations. Given $z\in Z$ the $U(1)$-orbit containing $z$ is the set $\mathcal{O}_z = \{\Phi_\theta(z)\mid \theta\in U(1)\}$ comprising all possible rotations of $z$. The vector field $R(z) = (\partial_\theta\Phi_\theta(z))|_{\theta = 0} $ is the infinitesimal generator of the $U(1)$-action, which is tangent to the collection of $U(1)$-orbits. We say that $X$ is $U(1)$-invariant with respect to a $U(1)$-action $\Phi_\theta$ if $F_t\circ\Phi_\theta = \Phi_\theta\circ F_t$ for every $\theta\in U(1)$ and $t\in\mathbb{R}$. If the $U(1)$-action is contextually clear, we will simply say $X$ is $U(1)$-invariant.

Fix a $U(1)$-action $\Phi_\theta$ and suppose that $X$ is $U(1)$-invariant. Assume there is a global Poincar\'e section $\Sigma\subset Z$ for $\Phi_\theta$ . Then, by definition of global Poincar\'e sections, $\Sigma$ is a hypersurface in $Z$ and each $U(1)$-orbit intersects $\Sigma$ uniquely and transversally. Let $\pi:Z\rightarrow \Sigma$ denote the mapping that sends $z\in Z$ to the unique point of intersection between $\Sigma$ and $\mathcal{O}_z$. The main text refers to $\pi$ as the footpoint map. A basic result referred to in the main text shows that $\pi$ maps the ODE system $X$ on $Z$ to another ODE system $X_\Sigma$ (i.e. a vector field) on $\Sigma$. The main text refers to $X_\Sigma$ as the footpoint flow field.

\begin{theorem}\label{footpoint_flow_thm}
There is a unique vector field $X_\Sigma$ on $\Sigma$ such that, for every streamline $z(t)$ of $X$, $\sigma(t) = \pi(z(t))$ is a streamline of $X_\Sigma$.
\end{theorem}
\begin{proof}
%First a simple observation. For each $\sigma\in \Sigma$ the preimage $\pi^{-1}(\sigma)$ is the $R$-streamline containing $\sigma$, and is therefore diffeomorphic to the circle $U(1)$. %It follows that $\pi: P\rightarrow \Sigma$ defines a fiber bundle with typical fiber $U(1)$.

Commutativity of $F_t$ and $\Phi_\theta$ implies there is a $1$-parameter family of mappings $f_t:\Sigma\rightarrow\Sigma$, $t\in\mathbb{R}$, such that $\pi\circ F_t = f_t\circ\pi$. To see this, let $z_0\in \mathcal{O}_\sigma$ be an arbitrary point in the $U(1)$-orbit containing $\sigma\in \Sigma$. The image point $z(t) = F_t(z_0)$ is contained in the $U(1)$-orbit $\mathcal{O}_{\pi(z(t))}$. If $z_0^\prime\in \mathcal{O}_\sigma$ is any other point in the $U(1)$-orbit containing $\sigma$ then there is a $\theta\in U(1)$ such that $z_0^\prime = \Phi_\theta(z_0)$. By commutativity, the new image point $z^\prime(t) = F_t(z_0^\prime)$ is related to the previous one according to $z^\prime(t) = F_t(\Phi_\theta(z_0)) = \Phi_\theta(F_t(z_0)) = \Phi_\theta(z(t))$. Thus, $z^\prime(t)$ and $z(t)$ lie on a common $U(1)$-orbit. The projected image $\pi(F_t(\mathcal{O}_{\sigma}))$ is therefore the singleton set $\{\sigma(t)\}$, where $\sigma(t) = \pi(F_t(z_0))$, for any $z_0\in \mathcal{O}_\sigma$. We define $f_t(\sigma) = \sigma(t)$. Let $z_0\in Z$ be any point in phase space and set $\sigma = \pi(z_0)$. By definition of $f_t$, we have
\begin{align*}
f_t(\pi(z_0))=f_t(\sigma) =\sigma(t) = \pi(F_t(z_0)),
\end{align*}
as claimed.

The commuting property $\pi\circ F_t = f_t\circ\pi$ implies the $1$-parameter family of mappings $f_t:\Sigma\rightarrow\Sigma$ satisfies the flow property $\forall t_1,t_2\in\mathbb{R},\quad f_{t_1+t_2} = f_{t_1}\circ f_{t_2}$. Indeed, if $\sigma\in\Sigma$ there is some $z\in Z$ with $\pi(z) = \sigma$, which implies
\begin{align*}
f_{t_1+t_2}(\sigma) = f_{t_1+t_2}(\pi(z)) &= \pi(F_{t_1+t_2}(z)) = \pi(F_{t_1}(F_{t_2}(z))) \\
&= f_{t_1}(\pi(F_{t_2}(z))) = f_{t_1}(f_{t_2}(\pi(z))) = f_{t_1}(f_{t_2}(\sigma)).
\end{align*}
Therefore $f_t$ is the flow map for a vector field $X_{\Sigma}$ on $\Sigma$. 

If $z(t)$ is any streamline for $X$ then $\sigma(t) = \pi(z(t))$ is a streamline for $X_\Sigma$. To see this first let $z_0 = z(0)$, $\sigma_0 = \sigma(0)$, and observe that $\sigma(t) = \pi(F_t(z_0)) = f_t(\sigma_0)$. Since $f_t$ is the flow map for $X_\Sigma$ it follows that $\sigma(t)$ is a streamline for $X_\Sigma$.

Suppose that $Y_\Sigma$ were a second vector field on $\Sigma$ sharing the previous property with $X_\Sigma$. For $\sigma_0\in \Sigma$ let $z_0$ be any point in the $U(1)$-orbit containing $\sigma_0$. There is a unique $X$-streamline $z(t)$ with $z(0) = z_0$. Moreover $\sigma(t) = \pi(z(t))$ is a streamline for both $X_\Sigma$ and $Y_\Sigma$. In particular,
\begin{align*}
Y_\Sigma(\sigma_0) = Y_\Sigma(\sigma(0)) = \frac{d}{dt}\bigg|_0\sigma(t) = X_\Sigma(\sigma(0)) = X_\Sigma(\sigma_0),
\end{align*}
which implies $X_\Sigma = Y_\Sigma$.
\end{proof}

\begin{remark}
This result has an interesting interpretation when $\Sigma$ is also a Poincar\'e section for $X$, the ODE system of interest, as happens in the case of charged particles moving in a strong magnetic field. Since $\Sigma$ is a Poincar\'e section for $X$ there is a well-defined poincar\'e map and associated discrete-time dynamics for $X$ on $\Sigma$. It is generally interesting to inquire as to whether there is a continuous-time dynamical system on $\Sigma$ that provides continuous-time interpolation of discrete-time Poincar\'e map dynamics. Theorem \ref{footpoint_flow_thm} says that, owing to the presence of $U(1)$-symmetry with a common Poincar\'e section, there is indeed such a dynamical system on $\Sigma$ -- that defined by the footpoint flow field $X_\Sigma$.
\end{remark}

\begin{remark}
Fig. \ref{sigma_drawing} gives a visual representation of Theorem \ref{footpoint_flow_thm}.
\end{remark}

\begin{figure}
\includegraphics[width=0.7\textwidth,trim={0 6cm 0 6cm}]{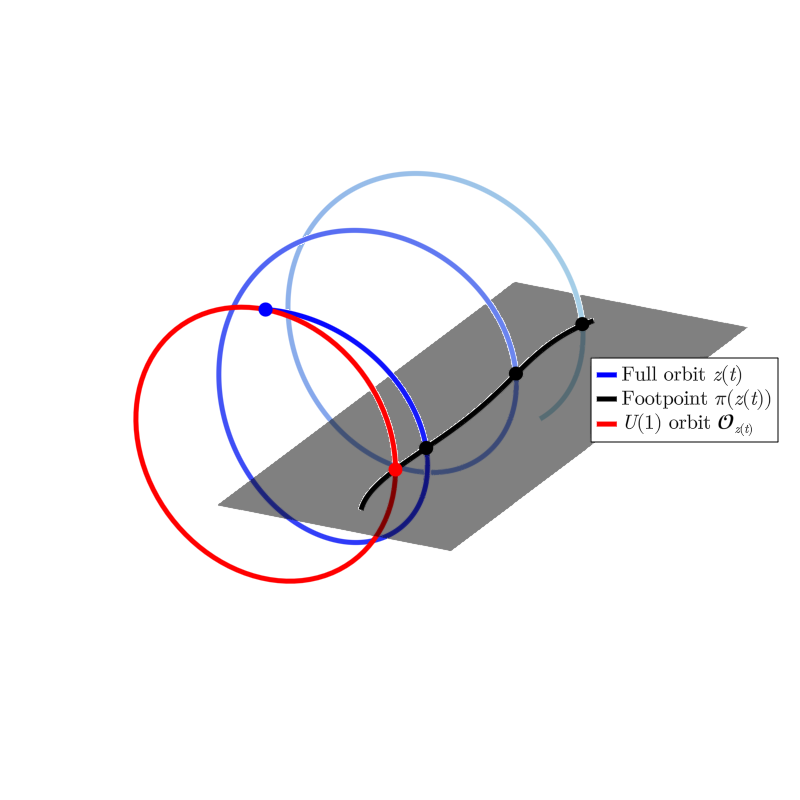}
\caption{\label{sigma_drawing} An illustration of a trajectory described by Theorem \ref{footpoint_flow_thm}.
The full orbit trajectory is given in blue and the corresponding footpoint trajectory is in black, traveling on the gray Poincar\'e section $\Sigma$.
An example of a $U(1)$-symmetry orbit is plotted in red. 
Although only a single $U(1)$-orbit is shown, there is such an orbit containing $z(t)$ for all times $t$.}
\end{figure}

Without further information about the ODE system $X$ finding its footpoint flow field requires detailed knowledge of the $U(1)$-orbits. These orbits are known when the underlying $U(1)$-symmetry is known in advance, but not when dealing with ``hidden" symmetries, as in non-perturbative guiding center modeling. Fortunately, matters simplify considerably when $X$ is a Hamiltonian system like the Lorentz force Law. The following Theorem shows that in the Hamiltonian setting the footpoint flow field is completely determined by $\Sigma$, the Poisson bracket for $X$, the Hamiltonian for $X$, and the conserved quantity $J$ associated with the $U(1)$-action $\Phi_\theta$; detailed knowledge of the $U(1)$-orbits is not required.

\begin{theorem}\label{hamiltonian_footpoint_thm}
If $X$ is a Hamiltonian system with Hamiltonian $H$ and Poisson bracket $\{\cdot,\cdot\}$ that satisfy
\begin{itemize}
\item $\{\cdot,\cdot\}$ and $H$ are each $U(1)$-invariant: for all $\theta\in U(1)$, $f,g:Z\rightarrow \mathbb{R}$, we have $H\circ\Phi_\theta = H$, $\{f\circ\Phi_\theta,g\circ\Phi_\theta\} = \{f,g\}\circ\Phi_\theta$,
\item the infinitesimal generator $R$ for the $U(1)$-action is a Hamiltonian vector field with Hamiltonian $J$,
%\item $\Sigma\subset Z$ is a Poincar\'e section for $R$,
\end{itemize} 
then the following is true.   
\begin{enumerate}
\item The hypersurface $\Sigma$ has a natural Poisson bracket $\{\cdot,\cdot\}_\Sigma$ with $J_\Sigma = J\mid\Sigma$ as a Casimir invariant. If $Z$ is equipped with coordinates $(\sigma^i,\zeta)$, where $\zeta$ is an angular coordinate, and $\Sigma = \{\zeta = 0\}$ then the $\sigma^i$ parameterize $\Sigma$. Moreover the Poisson bracket $\{f,g\}_{\Sigma}$ between functions $f = f(\sigma)$, 
 and $g = g(\sigma)$ on $\Sigma$ is given explicitly by
\begin{align}
\{f,g\}_\Sigma &= \partial_{\sigma}f^T\mathbb{J}_{\Sigma}\partial_{\sigma}g+\partial_\sigma J_\Sigma^T\mathbb{J}_\Sigma\,\partial_\sigma f\,\frac{N_\Sigma^T\partial_{\sigma}g}{N_\Sigma^T\partial_\sigma J_\Sigma}- \partial_\sigma J_\Sigma^T\mathbb{J}_\Sigma\,\partial_\sigma g\,\frac{N_\Sigma^T\partial_{\sigma}f}{N_\Sigma^T\partial_\sigma J_\Sigma}, \label{reduced_PB}
\end{align}
where the matrix $\mathbb{J}_\Sigma$ has components $\mathbb{J}_\Sigma^{ij} = \{\sigma^i,\sigma^j\}(\sigma,0)$, and the column vector $N$ has components $N^i_\Sigma = \{\sigma^i,\zeta\}(\sigma,0)$.
\item The footpoint flow field $X_{\Sigma}$ on $\Sigma$ is a Hamiltonian vector field with respect to the Poisson bracket $\{\cdot,\cdot\}_{\Sigma}$ and the Hamiltonian function $H_\Sigma = H\mid \Sigma$.
\end{enumerate}
\end{theorem}
\begin{proof}
Let $\{\cdot,\cdot\}$ denote the Poisson bracket on ambient phase space $Z$. If $F,G:Z\rightarrow\mathbb{R}$ are smooth functions then
\begin{align*}
\mathcal{L}_{R}\{F,G\} = \{\{F,G\},J\} = \{\{F,J\},G\} + \{F,\{G,J\}\} = \{\mathcal{L}_{R}F,G\} + \{F,\mathcal{L}_{R}G\},
\end{align*}
because $J$ is the Hamiltonian for $R$ and $\{\cdot,\cdot\}$ satisfies the Jacobi identity. In particular, if $F,G$ are both $U(1)$-invariant, so that $\mathcal{L}_RF = \mathcal{L}_RG = 0$, then so is their Poisson bracket, $\mathcal{L}_{R}\{F,G\} = 0$.

Let $\iota_{\Sigma}:\Sigma\rightarrow Z$ denote the inclusion map for the Poincar\'e section $\Sigma$. Given smooth functions $f,g:\Sigma\rightarrow\mathbb{R}$ define their bracket according to
\begin{align*}
\{f,g\}_{\Sigma} = \iota_{\Sigma}^*\{\pi^*f,\pi^*g\},
\end{align*}
where $\pi:Z\rightarrow\Sigma$ denotes the footpoint map. This bracket is clearly bilinear and skew-symmetric. It satisfies the Leibniz property because
\begin{align*}
\{f,gh\}_{\Sigma}  &= \iota_{\Sigma}^*\bigg(\{\pi^*f,\pi^*g\}\,\pi^*h + \{\pi^*f,\pi^*h\}\,\pi^*g\bigg)\\
&= \{f,g\}_{\Sigma}\,\iota_{\Sigma}^*\pi^*h + \{f,h\}_{\Sigma}\,\iota_{\Sigma}^*\pi^*g\\
& = \{f,g\}_{\Sigma}\,h + \{f,h\}_{\Sigma}\,g,
\end{align*}
where we have used the Leibniz property for $\{\cdot,\cdot\}$ and $\iota_{\Sigma}^*\pi^* = (\pi\circ\iota_{\Sigma})^* = \text{id}_{\Sigma}^* = 1$. It satisfies the Jacobi identity because (a) $\{\pi^*f,\pi^*g\}$ is $U(1)$-invariant since both $\pi^*f$ and $\pi^*g$ are $U(1)$-invariant, (b) $\pi^*\iota_{\Sigma}^*\{\pi^*f,\pi^*g\} = \{\pi^*f,\pi^*g\}$ since $\iota_{\Sigma}\circ\pi$ shifts points along $U(1)$-orbits, and (c) the Jacobi identity involves a cyclic sum of terms like
\begin{align*}
\{\{f,g\}_{\Sigma},h\}_{\Sigma} & = \iota_{\Sigma}^*\{\pi^*\iota_{\Sigma}^*\{\pi^*f,\pi^*g\},\pi^*h\} = \iota_{\Sigma}^*\{\{\pi^*f,\pi^*g\},\pi^*h\},
\end{align*}
which must vanish by the Jacobi identity for $\{\cdot,\cdot\}$. Thus, $\{\cdot,\cdot\}_{\Sigma}$ defines a Poisson bracket on $\Sigma$. The restriction of $J$ to $\Sigma$, $J_\Sigma = J\mid \Sigma = \iota_{\Sigma}^*J$, is a Casimir for $\{\cdot,\cdot\}_\Sigma$ because
\begin{align*}
\{f,J_\Sigma\}_\Sigma = \iota_{\Sigma}^*\{\pi^*f,\pi^*J_\Sigma\} = \iota_{\Sigma}^*\{\pi^*f,J\} = \iota_{\Sigma}^*(\mathcal{L}_{R}\pi^*f) = 0,
\end{align*}
for each smooth function $f$ on $\Sigma$.

The footpoint map $\pi:Z\rightarrow \Sigma$ is a Poisson map between $(Z,\{\cdot,\cdot\})$ and $(\Sigma,\{\cdot,\cdot\}_{\Sigma})$ because
\begin{align*}
\pi^*\{f,g\}_{\Sigma} = \pi^*\iota_{\Sigma}^*\{\pi^*f,\pi^*g\} = \{\pi^*f,\pi^*g\}.
\end{align*}
The Hamiltonian $H$ is $U(1)$-invariant by hypothesis and  therefore determined by its values on the Poincar\'e section, i.e. by the restriction $H_\Sigma = H\mid \Sigma$, in the sense that $\pi^*H_\Sigma = H$. It now follows from Guillemin-Sternberg collectivization \cite{Guillemin_Sternberg_1980} that if $z(t)$ is any solution of Hamilton's equations on $Z$ with Hamiltonian $H = \pi^*H_\Sigma$ then the corresponding footpoint trajectory $\sigma(t) = \pi(z(t))$ is a solution of Hamilton's equations on $\Sigma$ with Hamiltonian $H_\Sigma$. In other words, the footpoint trajectories are streamlines for the Hamiltonian vector field on $\Sigma$ with Hamiltonian $H_\Sigma$. But, by Theorem \ref{footpoint_flow_thm}, the only vector field on $\Sigma$ that has footpoint trajectories as streamlines is the footpoint flow field $X_\Sigma$. This establishes the second part of the theorem.

To complete the proof, we must now confirm the formula \eqref{reduced_PB} for the Poisson bracket $\{\cdot,\cdot\}_{\Sigma}$ in the coordinates $(\sigma^i,\zeta)$. Without loss of generality, suppose the index $i$ takes values in $\{1,\dots,d\}$, for some integer $d$. There must be a skew-symmetric matrix $\mathbb{J}^{ij}(\sigma,\zeta)$ and a column vector $N^i(\sigma,\zeta)$ such that
\begin{align*}
\{F,G\} &= \begin{pmatrix} \partial_{\sigma^1}F & \dots & \partial_{\sigma^d}F & \partial_\zeta F\end{pmatrix}\begin{pmatrix} \mathbb{J} & N\\ -N^T & 0\end{pmatrix}\begin{pmatrix} \partial_{\sigma^1}G\\ \vdots \\ \partial_{\sigma^d}G \\ \partial_{\zeta}G\end{pmatrix}\\
&=\partial_{\sigma}F^T\mathbb{J}\partial_{\sigma}G + \partial_\zeta G\,N^T\partial_{\sigma}F - \partial_\zeta F\,N^T\partial_{\sigma}G.
\end{align*}
Setting $F = \sigma^{i}$ and $G=\sigma^j$ shows that the components of $\mathbb{J}$ are given by
\begin{align*}
\mathbb{J}^{ij} = \{\sigma^i,\sigma^j\},
\end{align*}
while setting $F = \sigma^i$ and $G=\zeta$ shows that the components of $N$ are given by
\begin{align*}
N^i = \{\sigma^i,\zeta\}.
\end{align*}
The inclusion map $\iota_{\Sigma}:\Sigma\rightarrow Z$ is given by $\sigma\mapsto (\sigma,0)$ in these coordinates. The Poisson bracket $\{\cdot,\cdot\}_{\Sigma}$ is therefore given by
\begin{align}
\{f,g\}_{\Sigma} = \partial_{\sigma}f^T\mathbb{J}_{\Sigma}\partial_{\sigma}g + \partial_\zeta (\pi^*g)(\sigma,0)\,N_\Sigma^T\partial_{\sigma}f - \partial_\zeta (\pi^*f)(\sigma,0)\,N_\Sigma^T\partial_{\sigma}g,\label{proto_red_bracket}
\end{align}
where $\mathbb{J}_\Sigma = \mathbb{J}(\sigma,0)$ and $N_\Sigma = N(\sigma,0)$.
This would be an explicit formula for $\{f,g\}_{\Sigma}$ were it not for the appearance of the $\zeta$-derivatives $\partial_\zeta (\pi^*f)(\sigma,0)$ and $\partial_\zeta (\pi^*g)(\sigma,0)$. These $\zeta$-derivatives can be expressed in terms of derivatives of $f,g$ as follows. By definition of the footpoint map $\pi$, the functions $\pi^*f,\pi^*g$ are each constant along $U(1)$-orbits. Equivalently, $\mathcal{L}_{R}\pi^*f = \mathcal{L}_{R}\pi^*g =0$. Writing the infinitesimal generator in components as $R = R^\zeta\,\partial_\zeta + \bm{R}^i\,\partial_{\sigma^i}$ therefore implies $R^\zeta\,\partial_\zeta(\pi^*f) + \bm{R}^T\,\partial_\sigma(\pi^*f) = 0$, or
\begin{align*}
\partial_\zeta(\pi^*f) = -\frac{\bm{R}^T\,\partial_\sigma(\pi^*f)}{R^\zeta},\quad \partial_\zeta(\pi^*g) = -\frac{\bm{R}^T\,\partial_\sigma(\pi^*g)}{R^\zeta}.
\end{align*}
These expressions can be simplified further using the fact that $J$ is the Hamiltonian for $R$. In particular, the components of $R$ must be given in terms of derivatives of $J$ according to
\begin{align*}
R^\zeta = \{\zeta,J\} = -N^T\partial_\sigma J ,\quad \bm{R}  = \{\sigma,J\} = \mathbb{J} \partial_\sigma J + N\,\partial_\zeta J.
\end{align*}
The $\zeta$-derivatives are therefore
\begin{align*}
\partial_\zeta(\pi^*f) = \frac{(\mathbb{J} \partial_\sigma J + N\,\partial_\zeta J)^T\,\partial_\sigma(\pi^*f)}{N^T\partial_\sigma J},\quad \partial_\zeta(\pi^*g) = \frac{(\mathbb{J} \partial_\sigma J + N\,\partial_\zeta J)^T\,\partial_\sigma(\pi^*g)}{N^T\partial_\sigma J},
\end{align*}
and the Poisson bracket \eqref{proto_red_bracket} becomes
\begin{align*}
\{f,g\}_{\Sigma} &= \partial_{\sigma}f^T\mathbb{J}_{\Sigma}\partial_{\sigma}g\\
&+ \frac{(\mathbb{J}_\Sigma \partial_\sigma J_\Sigma + N_\Sigma\,\partial_\zeta J(\sigma,0))^T\,\partial_\sigma g}{N_\Sigma^T\partial_\sigma J_\Sigma}\,N_\Sigma^T\partial_{\sigma}f\\
&- \frac{(\mathbb{J}_\Sigma \partial_\sigma J_\Sigma + N_\Sigma\,\partial_\zeta J(\sigma,0))^T\,\partial_\sigma f}{N_\Sigma^T\partial_\sigma J_\Sigma}\,N_\Sigma^T\partial_{\sigma}g\\
%%%
&= \partial_{\sigma}f^T\mathbb{J}_{\Sigma}\partial_{\sigma}g+ \frac{(\mathbb{J}_\Sigma \partial_\sigma J_\Sigma )^T\,\partial_\sigma g}{N_\Sigma^T\partial_\sigma J_\Sigma}\,N_\Sigma^T\partial_{\sigma}f- \frac{(\mathbb{J}_\Sigma \partial_\sigma J_\Sigma )^T\,\partial_\sigma f}{N_\Sigma^T\partial_\sigma J_\Sigma}\,N_\Sigma^T\partial_{\sigma}g\\
%%%
&= \partial_{\sigma}f^T\mathbb{J}_{\Sigma}\partial_{\sigma}g+\partial_\sigma J_\Sigma^T\mathbb{J}_\Sigma\,\partial_\sigma f\,\frac{N_\Sigma^T\partial_{\sigma}g}{N_\Sigma^T\partial_\sigma J_\Sigma}- \partial_\sigma J_\Sigma^T\mathbb{J}_\Sigma\,\partial_\sigma g\,\frac{N_\Sigma^T\partial_{\sigma}f}{N_\Sigma^T\partial_\sigma J_\Sigma}.
\end{align*}
\end{proof}
\begin{remark}
For the purposes of non-perturbative guiding center modeling, $\Sigma$; the Poisson bracket for $X$; and the Hamiltonian for $X$; are each known explicitly in advance. Theorem \ref{hamiltonian_footpoint_thm} therefore implies that the conserved quantity $\mathcal{J}$ associated with the ``hidden" $U(1)$-symmetry of the Lorentz force Law is all that is needed to find an explicit formula for the footpoint flow field, and therefore the non-perturbatibe guiding center model. The underlying mechanism that enables this remarkable simplification is Noether's theorem. Finding a formula for the footpoint flow field in general requires detailed knowledge of the underlying $U(1)$-symmetry. On the other hand, Noether's theorem implies that any Hamiltonian $U(1)$-symmetry is completely determined by its corresponding conservation law. Thus, in the Hamiltonian setting complete knowledge of the Noether conserved quantity implies complete knowledge of the corresponding $U(1)$-symmetry and therefore complete knowledge of the footpoint flow field.
\end{remark}

When Theorem \ref{hamiltonian_footpoint_thm} is applied to the Lorentz force Law written in dimensionless variables as $\dot{\bm{v}} = \bm{v}\times\bm{B}(\bm{x})$, $\dot{\bm{x}} = \epsilon\,\bm{v}$ it leads to the following characterization of the non-perturbative guiding center equations of motion in general non-vanishing magnetic fields. This result shows explicitly how to remove the assumption $\bm{B} = B(x,y)\,\bm{e}_z$ used in the main text. It also recovers the main text's non-perturbative guiding center equations of motion when $\bm{B}$ is translation-invariant along $z$. 

%%%% begin thm 3
\begin{theorem}
Parameterize the guiding center Poincar\'e section $\Sigma$ using coordinates $\bm{\mathsf{X}} = \bm{x}$, $u = v_\parallel$, $r = v_\perp$. Let $J_\Sigma$ denote the restriction of the non-perturbative action integral to $\Sigma$. The footpoint flow field $X_{\Sigma} = \dot{\bm{\mathsf{X}}}\cdot\partial_{\bm{\mathsf{X}}} + \dot{u}\,\partial_u + \dot{r}\,\partial_r$ on $\Sigma$ for the Lorentz force system is given explicitly by  
\begin{align}
D_\epsilon\,\dot{\bm{\mathsf{X}}}\cdot \bm{b} & = \epsilon\,D_\epsilon\,u- {\epsilon\,\partial_u J_\Sigma \bigg( |\bm{B}|+ \epsilon\,{(u\,\bm{b} + r\,\bm{e}_1)\cdot[u\,r^{-1}\,\nabla\bm{b}\cdot\bm{e}_2 + \nabla\bm{e}_1\cdot\bm{e}_2]} \bigg)}\label{Xdotb}\\
%%%
D_\epsilon\,\dot{\bm{\mathsf{X}}}\cdot \bm{e}_1 & = \epsilon\,D_\epsilon\,r- \epsilon\,{\partial_r J_\Sigma \bigg( |\bm{B}|+ \epsilon\,{(u\,\bm{b} + r\,\bm{e}_1)\cdot[u\,r^{-1\,}\nabla\bm{b}\cdot\bm{e}_2 + \nabla\bm{e}_1\cdot\bm{e}_2]} \bigg)}\label{XdotOne}\\
%%%
D_\epsilon\,\dot{\bm{\mathsf{X}}}\cdot \bm{e}_2& =  -\epsilon^2\,{r^{-1}(u\,\bm{b} + r\,\bm{e}_1)\cdot\partial_{\bm{\mathsf{X}}}J_\Sigma + \epsilon^2\,r^{-1}(u\,\bm{b}+r\,\bm{e}_1)\cdot(\nabla\bm{b}\cdot\bm{e}_1) (r\,\partial_uJ_\Sigma - u\,\partial_rJ_\Sigma)}\label{XdotTwo}\\
%%%
D_\epsilon\,\dot{u} & = \epsilon\,\bigg(\bm{B}+ \epsilon\,(\bm{b}\cdot\nabla\times\bm{b})[u\,\bm{b} + r\,\bm{e}_1]\nonumber\\
&\hspace{3em}+\epsilon\,\bm{b}\times([u\,\bm{b} + r\,\bm{e}_1]\cdot\nabla\bm{b}) - \epsilon\,r\,\bm{e}_2\times (\nabla\bm{e}_1\cdot\bm{e}_2)\bigg)\cdot\partial_{\bm{\mathsf{X}}}J_\Sigma\label{udot_thm}\\
%%%
D_\epsilon\,\dot{r} & =  -\epsilon\,u\,r^{-1}\,\bigg(\bm{B}+ \epsilon\,(\bm{b}\cdot\nabla\times\bm{b})[u\,\bm{b} + r\,\bm{e}_1]\nonumber\\
&\hspace{6em}+\epsilon\,\bm{b}\times([u\,\bm{b} + r\,\bm{e}_1]\cdot\nabla\bm{b}) - \epsilon\,r\,\bm{e}_2\times (\nabla\bm{e}_1\cdot\bm{e}_2)\bigg)\cdot\partial_{\bm{\mathsf{X}}}J_\Sigma\label{rdot_thm},
\end{align}
where the denominator $D$ is given by
\begin{align}
D_\epsilon &=\epsilon\,r^{-1}\bm{e}_2\cdot\partial_{\bm{\mathsf{X}}} J_\Sigma + \epsilon\,r^{-1}([u\,\bm{b}+r\,\bm{e}_1]\cdot\nabla\bm{b}\cdot\bm{e}_1)\,\partial_u J_\Sigma +  r^{-1}|\bm{B}|\partial_r J_\Sigma\nonumber\\
&+ \epsilon\,r^{-1}\,(\bm{b}\cdot\nabla\times\bm{b})(u\partial_rJ_\Sigma - r\partial_uJ_\Sigma) + \epsilon\,(\partial_uJ_\Sigma\,\bm{b} + \partial_rJ_\Sigma\,\bm{e}_1)\cdot (\nabla\bm{e}_1\cdot\bm{e}_2).
% \frac{1}{r}\bm{e}_2\cdot\partial_{\bm{\mathsf{X}}}J_\Sigma &+ \left(\bm{e}_2\cdot\nabla\bm{b}\cdot\bm{e}_1 + \frac{u}{r}\bm{b}\cdot\nabla\bm{b}\cdot \bm{e}_2 + \bm{b}\cdot\nabla\bm{e}_1\cdot\bm{e}_2\right)\,\partial_uJ_\Sigma\nonumber\\
% &+ (|\bm{B}| + u\,\bm{b}\cdot\nabla\times\bm{b} + r\,\bm{e}_1\cdot\nabla\bm{e}_1\cdot\bm{e}_2)\frac{1}{r}\partial_rJ_\Sigma.
\end{align}
\end{theorem}
\begin{proof}
We will perform the calculation for the system $\dot{\bm{v}} = \bm{v}\times\bm{B}$, $\dot{\bm{x}} = \bm{v}$. The result in the Theorem statement follows from this calculation after applying the substitutions $t\rightarrow t/\epsilon$ and $\bm{B}\rightarrow \bm{B}/\epsilon$.

In coordinate-independent form, the Poisson bracket $\{\cdot,\cdot\}$ for the Lorentz force is given by
\begin{align*}
\{F,G\} =\partial_{\bm{x}}F\cdot\partial_{\bm{v}}G - \partial_{\bm{x}}G\cdot\partial_{\bm{v}}F +  \bm{B}\cdot \partial_{\bm{v}}F\times\partial_{\bm{v}}G.
\end{align*}
In terms of the coordinates $(\bm{x},v_\parallel,v_\perp,\zeta)$ defined by $\bm{v} = v_\parallel\bm{b} + v_\perp(\cos\zeta\,\bm{e}_1 + \sin\zeta\,\bm{e}_2)$, the bracket is instead
\begin{align*}
% \{F,G\} & = \bigg(\partial_{\bm{x}}F + v_\perp\,(\nabla\bm{b}\cdot\bm{a})\,\partial_{v_\parallel}F - v_\parallel\,(\nabla\bm{b}\cdot\bm{a})\,\partial_{v_\perp}F \bigg)\cdot\partial_{\bm{v}}G \\
% & + \partial_\zeta F\,\bigg( \frac{v_\parallel}{v_\perp}\,(\nabla\bm{b}\cdot \bm{c}) - (\nabla\bm{e}_2\cdot\bm{e}_1)\bigg)\cdot\partial_{\bm{v}}G\\
% &-\bigg(\partial_{\bm{x}}G + v_\perp\,(\nabla\bm{b}\cdot\bm{a})\,\partial_{v_\parallel}G - v_\parallel\,(\nabla\bm{b}\cdot\bm{a})\,\partial_{v_\perp}G \bigg)\cdot\partial_{\bm{v}}F \\
% & - \partial_\zeta G\,\bigg( \frac{v_\parallel}{v_\perp}\,(\nabla\bm{b}\cdot \bm{c}) - (\nabla\bm{e}_2\cdot\bm{e}_1)\bigg)\cdot\partial_{\bm{v}}F\\
% & +  \bm{B}\cdot (\partial_{v_\perp}F\,\bm{a} - v_\perp^{-1}\,\partial_\zeta F\,\bm{c})\times(\partial_{v_\perp}G\,\bm{a} - v_\perp^{-1}\,\partial_\zeta G\,\bm{c})\\
%%%
\{F,G\}& =\bm{b}\cdot (\partial_{\bm{x}}F\,\partial_{v_\parallel}G - \partial_{\bm{x}}G\,\partial_{v_\parallel}F)+\bm{a}\cdot (\partial_{\bm{x}}F\,\partial_{v_\perp}G - \partial_{\bm{x}}G\,\partial_{v_\perp}F)\\
& + (v_\perp\,\bm{a}\cdot\nabla\bm{b}\cdot\bm{a} + v_\parallel\,\bm{b}\cdot\nabla\bm{b}\cdot\bm{a})(\partial_{v_\parallel}F\,\partial_{v_\perp}G-\partial_{v_\parallel}G\,\partial_{v_\perp}F)\\
& + \bigg(N^{\bm{x}}\cdot \partial_{\bm{x}}F + N^{v_\parallel}\,\partial_{v_\parallel}F + N^{v_\perp}\,\partial_{v_\perp}F\bigg)\,\partial_\zeta G\\
&  - \bigg(N^{\bm{x}}\cdot \partial_{\bm{x}}G + N^{v_\parallel}\,\partial_{v_\parallel}G + N^{v_\perp}\,\partial_{v_\perp}G\bigg)\,\partial_\zeta F,
\end{align*}
where $\bm{a} = \cos\zeta\,\bm{e}_1 + \sin\zeta\,\bm{e}_2$, $\bm{c} = \bm{a}\times\bm{b} = -\cos\zeta\,\bm{e}_2+\sin\zeta\,\bm{e}_1$, and the components of $N$ are given by
\begin{align}
N^{\bm{x}} & = -\frac{1}{v_\perp}\bm{c}\\
N^{v_\parallel} & = -\bm{c}\cdot\nabla\bm{b}\cdot\bm{a} - \frac{v_\parallel}{v_\perp}(\bm{b}\cdot\nabla\bm{b})\cdot\bm{c} + \bm{b}\cdot\nabla\bm{e}_1\cdot\bm{e}_2\\
N^{v_\perp} & = \frac{|\bm{B}|}{v_\perp} + \frac{v_\parallel}{v_\perp}\,\bm{b}\cdot\nabla\times\bm{b} + \bm{a}\cdot\nabla\bm{e}_1\cdot\bm{e}_2.
\end{align}
Theorem \ref{hamiltonian_footpoint_thm} therefore implies that the Poisson bracket $\{\cdot,\cdot\}_\Sigma$ on the guiding center Poincar\'e section is given by,
\begin{align*}
\{f,g\}_\Sigma&=\bm{b}\cdot (\partial_{\bm{\mathsf{X}}}f\,\partial_{u}g - \partial_{\bm{\mathsf{X}}}g\,\partial_{u}f)+\bm{e}_1\cdot (\partial_{\bm{\mathsf{X}}}f\,\partial_{r}g - \partial_{\bm{\mathsf{X}}}g\,\partial_{r}f)\\
& + (r\,\bm{e}_1\cdot\nabla\bm{b}\cdot\bm{e}_1 + u\,\bm{b}\cdot\nabla\bm{b}\cdot\bm{e}_1)(\partial_{u}f\,\partial_{r}g-\partial_{u}g\,\partial_{r}f)\\
& +\bigg(\bm{b}\cdot (\partial_{\bm{\mathsf{X}}}J_\Sigma\,\partial_{u}f - \partial_{\bm{\mathsf{X}}}f\,\partial_{u}J_\Sigma)+\bm{e}_1\cdot (\partial_{\bm{\mathsf{X}}}J_\Sigma\,\partial_{r}f - \partial_{\bm{\mathsf{X}}}f\,\partial_{r}J_\Sigma)\\
& + (r\,\bm{e}_1\cdot\nabla\bm{b}\cdot\bm{e}_1 + u\,\bm{b}\cdot\nabla\bm{b}\cdot\bm{e}_1)(\partial_{u}J_\Sigma\,\partial_{r}f-\partial_{u}f\,\partial_{r}J_\Sigma)\bigg)\frac{N^{\bm{\mathsf{X}}}_\Sigma\cdot \partial_{\bm{\mathsf{X}}}g + N^u_{\Sigma}\,\partial_ug + N^r_\Sigma\,\partial_rg}{N^{\bm{\mathsf{X}}}_\Sigma\cdot \partial_{\bm{\mathsf{X}}}J_\Sigma + N^u_{\Sigma}\,\partial_uJ_\Sigma + N^r_\Sigma\,\partial_rJ_\Sigma}\\
%%%
& -\bigg(\bm{b}\cdot (\partial_{\bm{\mathsf{X}}}J_\Sigma\,\partial_{u}g - \partial_{\bm{\mathsf{X}}}g\,\partial_{u}J_\Sigma)+\bm{e}_1\cdot (\partial_{\bm{\mathsf{X}}}J_\Sigma\,\partial_{r}g - \partial_{\bm{\mathsf{X}}}g\,\partial_{r}J_\Sigma)\\
& + (r\,\bm{e}_1\cdot\nabla\bm{b}\cdot\bm{e}_1 + u\,\bm{b}\cdot\nabla\bm{b}\cdot\bm{e}_1)(\partial_{u}J_\Sigma\,\partial_{r}g-\partial_{u}g\,\partial_{r}J_\Sigma)\bigg)\frac{N^{\bm{\mathsf{X}}}_\Sigma\cdot \partial_{\bm{\mathsf{X}}}f + N^u_{\Sigma}\,\partial_uf + N^r_\Sigma\,\partial_rf}{N^{\bm{\mathsf{X}}}_\Sigma\cdot \partial_{\bm{\mathsf{X}}}J_\Sigma + N^u_{\Sigma}\,\partial_uJ_\Sigma + N^r_\Sigma\,\partial_rJ_\Sigma},
\end{align*}
where
\begin{align*}
N^{\bm{\mathsf{X}}}_\Sigma & =r^{-1}\bm{e}_2 \\
N^u_\Sigma & = r^{-1}(u\bm{b}+r\bm{e}_1)\cdot\nabla\bm{b}\cdot\bm{e}_2  - \bm{b}\cdot\nabla\times\bm{b}  + \bm{b}\cdot\nabla\bm{e}_1\cdot\bm{e}_2\\
N^r_\Sigma & = r^{-1}|\bm{B}| + r^{-1}u\,\bm{b}\cdot\nabla\times\bm{b} + \bm{e}_1\cdot\nabla\bm{e}_1\cdot\bm{e}_2.
\end{align*}
By Theorem \ref{hamiltonian_footpoint_thm}, the footpoint flow field for the Lorentz force is given by Hamilton's equations $\dot{\sigma}^i = \{\sigma^i,H_\Sigma\}_{\Sigma}$. The above explicit expression for $\{\cdot,\cdot\}_\Sigma$, together with $H_\Sigma = \frac{1}{2}u^2 + \frac{1}{2}r^2$, therefore implies the following explicit expressions for the components of the footpoint flow field. The $\bm{X}$-component is
\begin{align*}
\{\bm{\mathsf{X}},H_\Sigma\}_\Sigma&=u\,\bm{b}+r\,\bm{e}_1\\
& +\bigg(-\partial_{u}J_\Sigma\,\bm{b}-\partial_rJ_\Sigma\,\bm{e}_1\bigg)\frac{ N^u_{\Sigma}\,u + N^r_\Sigma\,r}{N^{\bm{\mathsf{X}}}_\Sigma\cdot \partial_{\bm{\mathsf{X}}}J_\Sigma + N^u_{\Sigma}\,\partial_uJ_\Sigma + N^r_\Sigma\,\partial_rJ_\Sigma}\\
%%%
& -\bigg((u\,\bm{b} + r\,\bm{e}_1)\cdot\partial_{\bm{\mathsf{X}}}J_\Sigma \\
& + (r\,\bm{e}_1\cdot\nabla\bm{b}\cdot\bm{e}_1 + u\,\bm{b}\cdot\nabla\bm{b}\cdot\bm{e}_1)(\partial_{u}J_\Sigma\,r-u\,\partial_{r}J_\Sigma)\bigg)\frac{N^{\bm{\mathsf{X}}}_\Sigma}{N^{\bm{\mathsf{X}}}_\Sigma\cdot \partial_{\bm{\mathsf{X}}}J_\Sigma + N^u_{\Sigma}\,\partial_uJ_\Sigma + N^r_\Sigma\,\partial_rJ_\Sigma},\\
%%%
& = u\,\bm{b}+r\,\bm{e}_1\\
&  - \frac{(\partial_{u}J_\Sigma\,\bm{b}+\partial_rJ_\Sigma\,\bm{e}_1)(N^u_{\Sigma}\,u + N^r_\Sigma\,r)  }{N^{\bm{\mathsf{X}}}_\Sigma\cdot \partial_{\bm{\mathsf{X}}}J_\Sigma + N^u_{\Sigma}\,\partial_uJ_\Sigma + N^r_\Sigma\,\partial_rJ_\Sigma}\\
&- \frac{ (u\,\bm{b} + r\,\bm{e}_1)\cdot\partial_{\bm{\mathsf{X}}}J_\Sigma\,N^{\bm{\mathsf{X}}}_\Sigma }{N^{\bm{\mathsf{X}}}_\Sigma\cdot \partial_{\bm{\mathsf{X}}}J_\Sigma + N^u_{\Sigma}\,\partial_uJ_\Sigma + N^r_\Sigma\,\partial_rJ_\Sigma}\\
&- \frac{ (r\,\bm{e}_1\cdot\nabla\bm{b}\cdot\bm{e}_1 + u\,\bm{b}\cdot\nabla\bm{b}\cdot\bm{e}_1)(\partial_{u}J_\Sigma\,r-u\,\partial_{r}J_\Sigma) N^{\bm{\mathsf{X}}}_\Sigma }{N^{\bm{\mathsf{X}}}_\Sigma\cdot \partial_{\bm{\mathsf{X}}}J_\Sigma + N^u_{\Sigma}\,\partial_uJ_\Sigma + N^r_\Sigma\,\partial_rJ_\Sigma}\\
%%%
& = u\,\bm{b}+r\,\bm{e}_1\\
&  - \frac{\partial_u J_\Sigma (u N^u_\Sigma + r\,N^r_\Sigma)  }{N^{\bm{\mathsf{X}}}_\Sigma\cdot \partial_{\bm{\mathsf{X}}}J_\Sigma + N^u_{\Sigma}\,\partial_uJ_\Sigma + N^r_\Sigma\,\partial_rJ_\Sigma}\,\bm{b}\\
&- \frac{ \partial_r J_\Sigma (u N^u_\Sigma + r\,N^r_\Sigma) }{N^{\bm{\mathsf{X}}}_\Sigma\cdot \partial_{\bm{\mathsf{X}}}J_\Sigma + N^u_{\Sigma}\,\partial_uJ_\Sigma + N^r_\Sigma\,\partial_rJ_\Sigma}\,\bm{e}_1\\
&- \frac{ (\frac{u}{r}\bm{b} + \bm{e}_1)\cdot\partial_{\bm{\mathsf{X}}}J_\Sigma + (\bm{e}_1\cdot\nabla\bm{b}\cdot\bm{e}_1 + \frac{u}{r}\,\bm{b}\cdot\nabla\bm{b}\cdot\bm{e}_1)(r\,\partial_uJ_\Sigma - u\,\partial_rJ_\Sigma) }{N^{\bm{\mathsf{X}}}_\Sigma\cdot \partial_{\bm{\mathsf{X}}}J_\Sigma + N^u_{\Sigma}\,\partial_uJ_\Sigma + N^r_\Sigma\,\partial_rJ_\Sigma}\,\bm{e}_2\\
%%%
& = u\,\bm{b}+r\,\bm{e}_1\\
&  - \frac{\partial_u J_\Sigma (|\bm{B}| + u\,\bm{b}\cdot\nabla\times\bm{b} + u\,\bm{e}_2\cdot\nabla\bm{b}\cdot \bm{e}_1 + \frac{u^2}{r}\bm{b}\cdot\nabla\bm{b}\cdot\bm{e}_2 + (u\bm{b} + r\,\bm{e}_1)\cdot\nabla\bm{e}_1\cdot\bm{e}_2 )  }{N^{\bm{\mathsf{X}}}_\Sigma\cdot \partial_{\bm{\mathsf{X}}}J_\Sigma + N^u_{\Sigma}\,\partial_uJ_\Sigma + N^r_\Sigma\,\partial_rJ_\Sigma}\,\bm{b}\\
&- \frac{ \partial_r J_\Sigma (|\bm{B}| + u\,\bm{b}\cdot\nabla\times\bm{b} + u\,\bm{e}_2\cdot\nabla\bm{b}\cdot \bm{e}_1 + \frac{u^2}{r}\bm{b}\cdot\nabla\bm{b}\cdot\bm{e}_2 + (u\bm{b} + r\,\bm{e}_1)\cdot\nabla\bm{e}_1\cdot\bm{e}_2 ) }{N^{\bm{\mathsf{X}}}_\Sigma\cdot \partial_{\bm{\mathsf{X}}}J_\Sigma + N^u_{\Sigma}\,\partial_uJ_\Sigma + N^r_\Sigma\,\partial_rJ_\Sigma}\,\bm{e}_1\\
&- \frac{ (\frac{u}{r}\bm{b} + \bm{e}_1)\cdot\partial_{\bm{\mathsf{X}}}J_\Sigma + (\bm{e}_1\cdot\nabla\bm{b}\cdot\bm{e}_1 + \frac{u}{r}\,\bm{b}\cdot\nabla\bm{b}\cdot\bm{e}_1)(r\,\partial_uJ_\Sigma - u\,\partial_rJ_\Sigma) }{N^{\bm{\mathsf{X}}}_\Sigma\cdot \partial_{\bm{\mathsf{X}}}J_\Sigma + N^u_{\Sigma}\,\partial_uJ_\Sigma + N^r_\Sigma\,\partial_rJ_\Sigma}\,\bm{e}_2,
\end{align*}
which reproduces Eqs.\,\eqref{Xdotb}-\eqref{XdotTwo} in the Theorem statement. The $u$-component is
\begin{align*}
\{u,H_\Sigma\}_\Sigma & =\\
& + r\,(r\,\bm{e}_1\cdot\nabla\bm{b}\cdot\bm{e}_1 + u\,\bm{b}\cdot\nabla\bm{b}\cdot\bm{e}_1)\\
& +\bigg(\bm{b}\cdot \partial_{\bm{\mathsf{X}}}J_\Sigma - (r\,\bm{e}_1\cdot\nabla\bm{b}\cdot\bm{e}_1 + u\,\bm{b}\cdot\nabla\bm{b}\cdot\bm{e}_1)\partial_{r}J_\Sigma\bigg)\frac{ N^u_{\Sigma}\,u + N^r_\Sigma\,r}{N^{\bm{\mathsf{X}}}_\Sigma\cdot \partial_{\bm{\mathsf{X}}}J_\Sigma + N^u_{\Sigma}\,\partial_uJ_\Sigma + N^r_\Sigma\,\partial_rJ_\Sigma}\\
%%%
& -\bigg((u\,\bm{b} + r\,\bm{e}_1)\cdot\partial_{\bm{\mathsf{X}}}J_\Sigma \\
& + (r\,\bm{e}_1\cdot\nabla\bm{b}\cdot\bm{e}_1 + u\,\bm{b}\cdot\nabla\bm{b}\cdot\bm{e}_1)(\partial_{u}J_\Sigma\,r-u\,\partial_{r}J_\Sigma)\bigg)\frac{ N^u_{\Sigma}}{N^{\bm{\mathsf{X}}}_\Sigma\cdot \partial_{\bm{\mathsf{X}}}J_\Sigma + N^u_{\Sigma}\,\partial_uJ_\Sigma + N^r_\Sigma\,\partial_rJ_\Sigma}\\
%%%%%%%%%%%%%
& = \frac{r\,(r\bm{e}_1\cdot\nabla\bm{b}\cdot\bm{e}_1 + u\bm{b}\cdot\nabla\bm{b}\cdot\bm{e}_1)\,N^{\bm{\mathsf{X}}}_\Sigma\cdot\partial_{\bm{\mathsf{X}}}J_\Sigma}{D}\\
& +\frac{r(N^r_\Sigma\,\bm{b} - N^u_\Sigma\,\bm{e}_1)\cdot \partial_{\bm{\mathsf{X}}}J_\Sigma}{D}\\
%%%%%%%%%%%%%%
& =\frac{(|\bm{B}| + u\bm{b}\cdot\nabla\times\bm{b} + r\,\bm{e}_1\cdot\nabla\bm{e}_1\cdot\bm{e}_2)}{D}\bm{b}\cdot\partial_{\bm{\mathsf{X}}}J_\Sigma\\
& - \frac{(r\bm{e}_1\cdot\nabla\bm{b}\cdot\bm{e}_2 - r\bm{b}\cdot\nabla\times\bm{b} + u\bm{b}\cdot\nabla\bm{b}\cdot\bm{e}_2 + r\bm{b}\cdot\nabla\bm{e}_1\cdot\bm{e}_2)}{D}\bm{e}_1\cdot\partial_{\bm{\mathsf{X}}}J_\Sigma\\
&+\frac{(r\,\bm{e}_1\cdot\nabla\bm{b}\cdot\bm{e}_1 + u\bm{b}\cdot\nabla\bm{b}\cdot\bm{e}_1)}{D}\bm{e}_2\cdot\partial_{\bm{\mathsf{X}}}J_\Sigma\\
%%%%%%%%%%%%%%%%
 & = \frac{\bm{B}+ (\bm{b}\cdot\nabla\times\bm{b})[u\,\bm{b} + r\,\bm{e}_1]+\bm{b}\times([u\bm{b} + r\bm{e}_1]\cdot\nabla\bm{b}) - r\,\bm{e}_2\times (\nabla\bm{e}_1\cdot\bm{e}_2)}{D}\cdot\partial_{\bm{\mathsf{X}}}J_\Sigma,
\end{align*}
which reproduces Eq.\,\eqref{udot_thm}. Finally, the $r$-component is
\begin{align*}%% f = r
\{r,H_\Sigma\}_\Sigma&= -u(r\,\bm{e}_1\cdot\nabla\bm{b}\cdot\bm{e}_1 + u\,\bm{b}\cdot\nabla\bm{b}\cdot\bm{e}_1)\\
& +\bigg(\bm{e}_1\cdot (\partial_{\bm{\mathsf{X}}}J_\Sigma)\\
& + (r\,\bm{e}_1\cdot\nabla\bm{b}\cdot\bm{e}_1 + u\,\bm{b}\cdot\nabla\bm{b}\cdot\bm{e}_1)(\partial_{u}J_\Sigma)\bigg)\frac{ N^u_{\Sigma}\,u + N^r_\Sigma\,r}{N^{\bm{\mathsf{X}}}_\Sigma\cdot \partial_{\bm{\mathsf{X}}}J_\Sigma + N^u_{\Sigma}\,\partial_uJ_\Sigma + N^r_\Sigma\,\partial_rJ_\Sigma}\\
%%%
& -\bigg((u\,\bm{b} + r\,\bm{e}_1)\cdot\partial_{\bm{\mathsf{X}}}J_\Sigma \\
& + (r\,\bm{e}_1\cdot\nabla\bm{b}\cdot\bm{e}_1 + u\,\bm{b}\cdot\nabla\bm{b}\cdot\bm{e}_1)(\partial_{u}J_\Sigma\,r-u\,\partial_{r}J_\Sigma)\bigg)\frac{ N^r_\Sigma}{N^{\bm{\mathsf{X}}}_\Sigma\cdot \partial_{\bm{\mathsf{X}}}J_\Sigma + N^u_{\Sigma}\,\partial_uJ_\Sigma + N^r_\Sigma\,\partial_rJ_\Sigma}\\
%%%%%%%%%%%%%%%%%
&= -\frac{u\,(r\,\bm{e}_1\cdot\nabla\bm{b}\cdot\bm{e}_1 + u\,\bm{b}\cdot\nabla\bm{b}\cdot\bm{e}_1)\,N^{\bm{\mathsf{X}}}_\Sigma}{D}\cdot \partial_{\bm{\mathsf{X}}}J_\Sigma \\
& + \frac{u\,N^u_\Sigma\,\bm{e}_1 - u\,N^r_\Sigma\,\bm{b}}{D}\cdot\partial_{\bm{\mathsf{X}}}J_\Sigma\\
%%%%%%%%%%%%%%%%%%
&= \frac{-\frac{u}{r}\bm{b}\times([u\bm{b}+r\bm{e}_1]\cdot\nabla\bm{b}) - \frac{u}{r}[u\bm{b}+r\bm{e}_1](\bm{b}\cdot\nabla\times\bm{b}) + u\bm{e}_2\times(\nabla\bm{e}_1\cdot\bm{e}_2)}{D}\cdot\partial_{\bm{\mathsf{X}}}J_\Sigma,
\end{align*}
which reproduced Eq.\,\eqref{rdot_thm} from the Theorem statement.

\end{proof}
%%%%%%% end thm 3

\subsection{Derivative estimation}
This Section describes the method we used to estimate ``ground truth" time derivatives of (presumptive) continuous-time trajectories on $\Sigma$ that interpolate iterates of the Poincar\'e map. To provide good input data to estimate $\mathcal{J}$, we need to be confident that we can both trust that a point $(x,y,r)$ is integrable and that we have a good estimate of the derivative $(\dot{x},\dot{y})$.
For both of these tasks, we leverage the Birkhoff Reduced Rank Extrapolation (Birkhoff RRE) algorithm \cite{Ruth:2024} implemented in the \texttt{SymplecticMapTools.jl} Julia package. 
Using a single trajectory from the Poincar\'e map, here denoted $F$, Birkhoff RRE first classifies that trajectory as an invariant circle, an island, or chaos.
Then, assuming the trajectory is an invariant circle with Diophantine rotation number $\omega$, it returns an approximate parameterization of an invariant circle $z : \mathbb{T} \to \Sigma$ and $\omega$ so that $F(z(\theta)) = z(\theta+\omega)$.
Clearly, the adiabatic invariant $\mathcal{J}$ should be constant on $z$, meaning that $\dot{z} \cdot \nabla \mathcal{J} = 0$, giving the correct direction of the derivatives of $\dot{z}$ in 2D. 

To obtain the magnitude of of the derivatives, we first observe that the dynamics predicted by the non-perturbative guiding center equations of motion on the invariant circle will have the form $(\dot{x},\dot{y}) = ({d z}/{d\theta})\dot{\theta}$. 
To estimate $\dot{\theta}$, we note it must satisfy a differential equation of the form
\begin{equation*}
    \dot{\theta} = f(\theta), \qquad \theta(t+T(\theta(t))) = \theta(t)+\omega,
\end{equation*}
where $T(\theta(t))$ is the gyroperiod at the point $z(\theta(t))$. 
We can lift and invert the dynamics to find
\begin{equation*}
    \frac{dt}{d\theta} = g(t), \qquad t(\theta+\omega) = t(\theta)+T(\theta). 
\end{equation*}
For Diophantine $\omega$ and smooth enough $T$, we can solve the right expression for $t$ to find
\begin{equation*}
    t(\theta) = c + t_0 \theta + \sideset{}{'}\sum_{n=-N}^N t_n e^{i n \theta},
\end{equation*}
where $c$ is a constant determined by initial conditions and 
\begin{equation*}
    t_n = \begin{cases} \frac{1}{\omega} \langle T \rangle & n=0 \\ \frac{1}{e^{in\omega}-1} \langle T(\cdot) e^{-i n \cdot} \rangle & n\neq 0 \end{cases}, \quad \langle f \rangle = \frac{1}{2\pi} \int_{\mathbb{T}} f(\theta) \, d\theta.
\end{equation*}
Using a discrete Fourier transform to compute the above quantities and $\dot{\theta} = ({dt}/{d\theta})^{-1}$, we have an approximation of $(\dot{x},\dot{y})$.

%%%% supplementary material

%\bibliography{/Users/josh/Dropbox/Apps/Texpad/latex/cumulative_bib_file.bib}
\bibliographystyle{unsrt}
%\bibliography{cumulative_bib_file.bib}
%\input{manuscript_final_revision.bbl}
%%%%%%%%%%%%%%%%%%%%%%%%%%%%%%%%%%%%%

\providecommand{\noopsort}[1]{}\providecommand{\singleletter}[1]{#1}%

%%%%%%%%%%%%%%%%%%%%%%%%%%%%%%%%%%%%

\end{document}